\documentclass[11pt]{article}
\usepackage[a4paper,centering,scale=0.75]{geometry}
\usepackage{inputenc}
\usepackage{mathtools}
\usepackage{amsthm}
\usepackage{amsfonts,amssymb,mathrsfs}
\usepackage{bbm}
\usepackage[nofiglist, notablist]{endfloat}
\usepackage{booktabs}

\newtheorem{thm}{Theorem}[section]
\newtheorem{lemma}[thm]{Lemma}
\newtheorem{prop}[thm]{Proposition}

\theoremstyle{definition}

\theoremstyle{remark}
\newtheorem{rmk}[thm]{Remark}

\newcommand{\E}{\mathbb{E}}
\newcommand{\cF}{\mathscr{F}}
\newcommand{\enne}{\mathbb{N}}
\renewcommand{\P}{\mathbb{P}}
\newcommand{\erre}{\mathbb{R}}

\newcommand{\ind}[1]{\mathbbm{1}_{#1}}

\DeclarePairedDelimiter\abs{\lvert}{\rvert}
\DeclarePairedDelimiter\norm{\lVert}{\rVert}
\DeclarePairedDelimiterX\ip[2]{\langle}{\rangle}{#1,#2}

\begin{document}

\title{Nonparametric estimates of option prices via\\
  Hermite basis functions} 

\author{Carlo Marinelli\thanks{Department of Mathematics, University
    College London, Gower Street, London WC1E 6BT, UK.} \and Stefano
  d'Addona\thanks{Dipartimento di Scienze Politiche, Universit\`a di
    Roma Tre, Via G.~Chiabrera 199, 00145 Rome, Italy.}}

\date{\normalsize August 12, 2023}
\maketitle

\begin{abstract}
  We consider approximate pricing formulas for European options based
  on approximating the logarithmic return's density of the underlying
  by a linear combination of rescaled Hermite polynomials. The
  resulting models, that can be seen as perturbations of the classical
  Black-Scholes one, are nonpararametric in the sense that the
  distribution of logarithmic returns at fixed times to maturity is
  only assumed to have a square-integrable density. We extensively
  investigate the empirical performance, defined in terms of
  out-of-sample relative pricing error, of this class of approximating
  models, depending on their order (that is, roughly speaking, the
  degree of the polynomial expansion) as well as on several ways to
  calibrate them to observed data. Empirical results suggest that such
  approximate pricing formulas, when compared with simple
  nonparametric estimates based on interpolation and extrapolation on
  the implied volatility curve, perform reasonably well only for
  options with strike price not too far apart from the strike prices
  of the observed sample.
\end{abstract}

\section{Introduction}
Our aim is to construct approximate pricing formulas for European
options with fixed time to maturity by series expansion of return
distributions, to discuss their implementation, and to test their
empirical accuracy. The approach is nonparametric in the sense that we
do not make any parametric assumption on the distribution of
returns. Such distribution is instead approximated by (the integral
of) a truncated series of suitably weighted and scaled Hermite
polynomials, in such a way that the zeroth-order approximation
coincides with the standard Black-Scholes model.

Let $(\Omega,\cF,\P)$ be a probability space endowed with a filtration
$(\cF_t)_{t \geq 0}$, on which all random elements will be
defined. Let $\widehat{S}\colon \Omega \times \erre_+ \to \erre$ be
the adapted price process of a dividend-paying asset, with adapted
dividend process $q\colon \Omega \times \erre_+ \to \erre$ such that
$\exp\bigl(\int_0^t q_s\,ds\bigr)$ is bounded for every
$t \in \erre_+$, and let
$\widehat{Y} \colon \Omega \times \erre_+ \to \erre$ be the
corresponding adapted yield process defined by
\[
  \widehat{Y}_t := S_t + \int_0^t q_s\widehat{S}_s\,ds
  \qquad \forall t \in \erre_+.
\]
Denoting by $\beta$ the adapted continuous strictly positive price
process of the riskless cash account, i.e.
\[
  \beta_t = \exp\biggl( \int_0^t r_s\,ds \biggr) \qquad \forall t \in
  \erre_+,
\]
with $r$ the risk-free rate, and by $S:=\beta^{-1}\widehat{S}$ the
discounted asset price, we assume that the discounted yield process
$Y$ defined by
\[
Y_t = S_t + \int_0^t q_sS_s\,ds \qquad \forall t \in \erre_+
\]
is a martingale. In other words, we assume that $\P$ is a martingale
measure.  The martingale property of $Y$ implies that the process
$S\exp\bigl( \int_0^\cdot q_s\,ds \bigr)$ is also a martingale (see,
e.g., \cite{cm:JEF17} for details). In the Black-Scholes setting, one
assumes that there exists a constant $\sigma>0$ such that
\[
  S_t\exp\biggl( \int_0^t q_s\,ds \biggr) = S_0 \exp\biggl( \sigma W_t -
  \frac12 \sigma^2 t \biggr)
\]
for every $t \geq 0$, where $W$ is a standard Wiener
process. Therefore, denoting by $Z$ a standard Gaussian random
variable, one has
\[
  S_t\exp\biggl( \int_0^t q_s\,ds \biggr)
  = S_0 \exp\Bigl( \sigma\sqrt{t}
Z - \frac12 \sigma^2 t \Bigr)
\]
in law for every $t \geq 0$. Assuming for simplicity that $q$ is
constant, this implies
\begin{align*}
\E g(S_t) 
&= \E g\bigl( S_0 \exp \bigl( \sigma\sqrt{t}Z - \sigma^2 t/2 - qt \bigr)\\
&= \int_\erre g\bigl( S_0 \exp \bigl( \sigma\sqrt{t}x - \sigma^2 t/2
- qt \bigr) \phi(x)\,dx,
\end{align*}
where $\phi$ is the density of the standard Gaussian measure $\gamma$
on $\erre$ and $g\colon \erre \to \erre$ is any measurable function
either positive or such that the above integral is finite. As is well
known, this identity yields, as particular cases, the Black-Scholes
formula for put and call options\footnote{We consider only European
  options, and put and call options are always meant to be so-called
  vanilla options.}, as well as the Black-Scholes PDE (see, e.g.,
\cite{FoellSch}).

Empirical evidence suggests that observed option prices are not
compatible with such a model, and a vast literature exists about
alternative models that offer a better accuracy for pricing
purposes. A simple nonparametric approach consists in the following
steps: estimate the implied volatility from a set of observed call and
put option prices; view the implied volatility surface as a function
of (at least) the time to maturity and the strike price, say
$v\colon [t_0,t_1] \times [k_0,k_1] \to \erre_+$; given an option with
strike price $k$ and time to maturity $t$, obtain an estimate
$\hat{v}(t,k)$ of the corresponding volatility by interpolation on
$v$; obtain an estimate of the option price in the Black-Scholes
setting with volatility $\hat{v}(t,k)$. This is reasonable if the
point $(t,k)$ belongs to the convex envelope of the set of points
$(t_i,k_i)$ for which option prices are observed, and it was shown to
perform well in practice in \cite{cm:JEF17}.

Restricting to the case where the time to maturity $t$ is fixed,
another approach consists in the estimation of the law of the return
over the interval $[0,t]$, and to integrate with respect to the
estimated law to obtain estimates of option prices.
To fix ideas, discarding dividends for simplicity, one could write
$S_t = S_0 e^{R}$, where $R$ is the logarithmic return over $[0,t]$,
and, assuming that $R$ has a density $f^R$,
\[
\E g(S_t) = \int_\erre g\bigl( S_0 e^x \bigr) f^R(x)\,dx.
\]
In order to proceed in a nonparametric way, i.e. without assuming
that $f_R$ belongs to a family of density functions indexed by
finitely many parameters, a possibility is to assume that $f_R$
belongs to $L^2(\erre)$, to expand it as a series with respect to a
complete orthonormal basis, and to use as approximation a truncation
of the series to a finite sum. For instance, Lemma~\ref{lm:sigma}
below yields, for any parameters $m$ and $\sigma>0$,
\[
  f^R(x) = \sum_{n=0}^\infty \alpha_n h_n
  \Bigl( \frac{\sqrt{2}(x-m)}{\sigma} \Bigr)
  \exp\Bigl( -\frac{(x-m)^2}{2\sigma^2} \Bigr)
\]
as an identity of functions in $L^2(\erre)$, where $h_n$ is the $n$-th
Hermite polynomial and
\[
\alpha_n := \frac{1}{n!\,\sigma\sqrt{\pi}} \int_\erre f^R(x) 
h_n\bigl(\sqrt{2}(x-m)/\sigma\bigr)
e^{-(x-m)^2/2\sigma^2}\,dx.
\]
Introducing the random variable $X$ defined by $R = \sigma X + m$, so
that $S_t = S_0\exp(\sigma X+m)$, and denoting the density of $X$ by
$f$, this is equivalent to writing
\begin{equation}
  \label{eq:op}
  \E g(S_t) = \int g(S_0 e^{\sigma x+m}) f(x)\,dx
\end{equation}
and approximating $f$ by
\[
f_N(x) = \sum_{n=0}^N \alpha_n h_n(\sqrt{2}x) e^{-x^2/2},
\]
with
\[
\alpha_n := \frac{1}{n!\sqrt{\pi}} \int_\erre f(x) h_n(\sqrt{2}x)
e^{-x^2/2}\,dx.
\]
The coefficients $m, \sigma$ as well as $\alpha_0,\ldots,\alpha_N$ can
then be calibrated by minimizing a distance between observed prices
and ``approximate'' prices implied by replacing $f$ with $f_N$ in equation
\eqref{eq:op}. Several procedures to achieve this are discussed in Section
\ref{sec:cal}.  Moreover, note that a zero-th order approximation of
$f$ reduces to Black-Scholes pricing, choosing
$\sigma=\sigma_0\sqrt{t}$ and $m=-\sigma_0^2 t/2$, with $\sigma_0$ the
volatility of the underlying. Expansion in Hermite polynomials have
already been used to approximate densities of financial returns (with
fixed time) in diffusion and jump-diffusion models (see,~e.g.,
\cite{Xiu} and references therein), but we are not aware of any
previous work where the natural nonparametric Ansatz proposed here is
studied. On the other hand, a somewhat related, short study using
simulated index prices and other families of orthogonal polynomials
can be found in \cite{GriHaeSchi}.

Our main interest is to test the empirical performance of the
approximate pricing approach described above, dubbed Hermite pricing
for convenience, investigating its dependence on several factors, such
as the number of Hermite polynomials used, the calibration procedure,
the corresponding optimization algorithm, and so on. We shall consider
as benchmark a simple nonparametric pricing technique based on the
Black-Scholes model and interpolation on the implied volatility
curve. It was shown in \cite{cm:JEF17} that this simple method
outperforms more sophisticated techniques based on estimating the
density of logarithmic returns by second derivative of call prices
with respect to strike results (cf.,~e.g., \cite{AL}), as well as some
parametric methods. It seems therefore sufficient to use just this
simple technique as term of comparison.

The extensive empirical study conducted here suggests that Hermite
pricing performs reasonably well for options with strike price not too
far away from the strike prices of observed option prices, and is
quite unreliable otherwise. This appears to be the case across all
calibration methods used, even though some techniques are more robust
than others.  Such an observation is certainly not surprising: most
(nonparametric) methods generally suffer, roughly speaking, of poor
performance on points that lie outside the convex hull of the observed
data set, or, more generally, on regions of the data set that are
``sparsely populated''. On the other hand, Hermite polynomials are
defined on the whole real line, so once the coefficients of a linear
combination of them are estimated, the model can in principle produce
estimates for any data points, without resorting on extrapolation, as
is the case for the elementary implied volatility method already
mentioned. Even on rather rich data sets, however, estimates obtained
by Hermite pricing are often unreliable. In essence, we believe that
one can reasonably conclude that Hermite pricing can usefully
complement other pricing techniques, but it is not a plausible tool to
price (nonparametrically) ``outside the convex hull'' of observed
data points.
The qualitative results of the empirical analysis on real data are
essentially confirmed by an analogous statistical exercise conducted
on a smaller synthetic dataset generated using (non-Gaussian) Hermite
processes.

The content of the remaining part of the text is organized as follows:
in Section \ref{sec:prel} we collect some facts about Hermite polynomials,
compute some integrals with respect to Gaussian measures (on the real
line), and we recall the connection among European option prices,
distributions of returns, and implied volatility.
Pricing estimates for put options, essentially in closed form, implied
by approximating the density of returns with finite linear
combinations of rescaled Hermite polynomials are discussed in
Section \ref{sec:Hp}. Corresponding formulas for call options can be
formally obtained in a very similar way, but the payoff function of
put options is bounded, while the payoff function of call options is
not. For this (technical) reason we concentrate on the case of put
options.
The important issue of calibration is discusses in
Section \ref{sec:cal}. The main criterion is the minimization of the
relative pricing error, both in $\ell^2$ and in $\ell^1$ sense
(corresponding to least squares and least absolute deviation,
respectively). While the objective functions are smooth, in general
there is no convexity, so global optimization is hard. Explicit
expressions for the minimum points cannot be obtained, so numerical
minimization is needed.
An extensive empirical analysis is carried out in Section
\ref{sec:emp}: for fixed time and time to maturity, we use a set of
option prices to calibrate the model, the pricing estimates of which
are in turn compared with actual option prices. 
A similar analysis is carried out on a set of synthetic data in
Section \ref{sec:emps}, where functionals obtained from the payoff
function of put options are applied to exponentials of simulated
Hermite processes.
Finally, auxiliary material is collected in the appendix.

\section{Preliminaries}
\label{sec:prel}
\subsection{Notation}
The usual Lebesgue spaces $L^p(\erre)$, $p \in [1,\infty]$, will
simply be denoted by $L^p$. The scalar product in $L^2$ will be
denoted by $\ip{\cdot}{\cdot}$. We shall use the same symbol for the
scalar product in other spaces, whenever it is clear from the context
what is meant. Given a countable set of indices $I$, we shall use
standard notation for the usual sequence spaces $\ell^p(I)$, defined
as the set of sequences $x = (x_i)_{i \in I}$ such that
\[
  {\norm{x}}_{\ell^p(I)} = \Bigl( \sum_{i \in I} \abs{x_i}^p \Bigr)^{1/p}
  < \infty.
\]
Whenever $I$ is omitted, it is either $\mathbb{Z}_+$ or a finite set
clear from the context.

\subsection{Gaussian measures and Hermite polynomials}
For any real numbers $m$ and $\sigma$, with $\sigma \neq 0$, let
$\gamma_{m,\sigma}$ denote the Gaussian measure on $\erre$ with mean
$m$ and variance $\sigma^2$, that is the measure having density with
respect to Lebesgue measure given by
\[
x \mapsto \frac{1}{\sigma\sqrt{2\pi}} e^{(s-m)^2/2\sigma^2}.
\]
If $m=0$ and $\sigma=1$, we shall just write $\gamma$ in place of
$\gamma_{0,1}$.

The Hermite polynomials $(h_n)_{n \geq 0}$, defined by
\[
h_n(x) := (-1)^n e^{x^2/2} \frac{d^n}{dx^n} e^{-x^2/2}, \qquad n=0,1,2,\ldots,
\]
form a complete orthogonal system of the Hilbert space
$L^2(\gamma)$. The first few of them are
\[
h_0(x)=1, \qquad h_1(x)=x, \qquad h_2(x)=x^2-1, \qquad h_3(x) = x^3-3x.
\]
Moreover, $((n!)^{-1/2}h_n)_{n \geq 0}$ is a complete orthonormal
basis of $L^2(\gamma)$ -- see, e.g., \cite{Mall} for details.

Simple calculations based on a change of variable immediately show
that the rescaled shifted Hermite polynomials $x \mapsto (n!)^{-1/2}
h_n(\sigma^{-1}(x-m))$ form a complete orthonormal system of the
Hilbert space $L^2(\gamma_{m,\sigma})$.

The following observations are elementary but important in the sequel.
\begin{lemma}
  \label{lm:triv}
  Let $m,\sigma \in \erre$, $\sigma>0$, and $g\colon \erre \to \erre$
  be a measurable function. One has
  \[
    \norm[\Big]{x \mapsto g(x)e^{(x-m)^2/4\sigma^2}}_{L^2(\gamma_{m,\sigma})} =
    {(\sigma^2 2\pi)}^{-1/4} \norm[\big]{g}_{L^2(\erre)}.
  \]
  In particular, the function $g$ belongs to $L^2(\erre)$ if and only
  if $x \mapsto g(x) e^{(x-m)^2/4\sigma^2}$ belongs to
  $L^2(\gamma_{m,\sigma})$.
\end{lemma}
\begin{proof}
  In fact,
  \[
  \int_\erre \abs{g(x)}^2\,dx = \int_\erre \bigl( g(x)
  e^{(x-m)^2/4\sigma^2} \bigr)^2 e^{-(x-m)^2/2\sigma^2}\,dx
  \]
  and
  \begin{align*}
  \norm[\Big]{x \mapsto g(x)e^{(x-m)^2/4\sigma^2}}^2_{L^2(\gamma_{m,\sigma})}
  &= \frac{1}{\sigma \sqrt{2\pi}} \int_\erre \bigl( 
    g(x) e^{(x-m)^2/4\sigma^2} \bigr)^2 e^{-(x-m)^2/2\sigma^2}\,dx\\
  &= {(\sigma^2 2\pi)}^{-1/2} \norm[\big]{g}^2_{L^2(\erre)}.
  \qedhere
  \end{align*}
\end{proof}

\begin{lemma}
  \label{lm:sigma}
  Let $m,\sigma \in \erre$, $\sigma>0$, and $g \in L^2(\erre)$. The
  sequence $(\alpha_n)_{n \geq 0}$ defined by
  \begin{align*}
    \alpha_n &:= \frac{1}{\sigma\sqrt{n!2\pi}} \int_{-\infty}^{+\infty}
               g(x) h_n(\sigma^{-1}(x-m)) e^{-(x-m)^2/4\sigma^2}\,dx\\
             &= \frac{1}{\sqrt{n!2\pi}} \int_{-\infty}^{+\infty}
               g(\sigma x + m) h_n(x) e^{-x^2/4}\,dx
  \end{align*}
  belongs to $\ell^2$ and is such that
  \[
  g(x) = \sum_{n=0}^\infty \alpha_n (n!)^{-1/2}
         h_n(\sigma^{-1}(x-m)) e^{-(x-m)^2/4\sigma^2}
  \]
  as an identity in $L^2(\erre)$. Moreover,
  \[
  \norm[\big]{(\alpha_n)}_{\ell^2} = {(\sigma^2 2\pi)}^{-1/4} 
  \norm[\big]{g}_{L^2(\erre)}.
  \]
\end{lemma}
\begin{proof}
  It follows by Lemma~\ref{lm:triv} that $x \mapsto g(x)
  e^{(x-m)^2/4\sigma^2} \in L^2(\gamma_{m,\sigma})$, hence, 
  by Parseval's identity,
  \[
    g(x) e^{(x-m)^2/4\sigma^2} = \sum_{n=0}^\infty \alpha_n (n!)^{-1/2}
    h_n(\sigma^{-1}(x-m))
  \]
  in $L^2(\gamma_{m,\sigma})$, where, with a slight but harmless abuse
  of notation,
  \begin{align*}
    \alpha_n
    &:= \ip[\Big]{g(x)e^{(x-m)^2/4\sigma^2}}%
                 {(n!)^{-1/2} h_n(\sigma^{-1}(x-m))}_{L^2(\gamma_{m,\sigma})}\\
    &= \frac{1}{\sigma\sqrt{n!2\pi}} \int_{-\infty}^{+\infty}
               g(x) h_n(\sigma^{-1}(x-m)) e^{-(x-m)^2/4\sigma^2}\,dx\\
    &= \frac{1}{\sqrt{n!2\pi}} \int_{-\infty}^{+\infty}
               g(\sigma x+m) h_n(x) e^{-x^2/4}\,dx.
  \end{align*}
  Lemma~\ref{lm:triv} then implies
  \[
    g(x) = \sum_{n=0}^\infty \alpha_n (n!)^{-1/2}
           h_n(\sigma^{-1}(x-m)) e^{-(x-m)^2/4\sigma^2}
  \]
  in $L^2(\erre)$ and
  \[
  \norm[\big]{(\alpha_n)}_{\ell^2} = \norm[\Big]{x \mapsto
    g(x)e^{(x-m)^2/4\sigma^2}}_{L^2(\gamma_{m,\sigma})} 
    = {(\sigma^2 2\pi)}^{-1/4} \norm[\big]{g}_{L^2(\erre)}.
    \qedhere
  \]
\end{proof}

\subsection{Integrals with respect to Gaussian measures}
\label{ssec:gpi}
We shall need some explicit Gaussian indefinite integrals. In
particular, for any $n \geq 1$, integration by parts gives the
identities
\begin{align*}
\int x^n e^{-x^2/2} \,dx&=\int x^{n-1} x e^{-x^2/2} \,dx \\
&=-x^{n-1} e^{-x^2/2} + (n-1)\int x^{n-2} e^{-x^2/2}\,dx + c
\end{align*}
(here and in the following $c \in \erre$ denotes a constant), from
which it follows that
\begin{align}
  \label{eq:gi2}
  \int x^2 e^{-x^2/2}\,dx &= -xe^{-x^2/2} + \int e^{-x^2/2}\,dx + c,\\
  \label{eq:gi3}
  \int x^3 e^{-x^2/2}\,dx &= -x^2e^{-x^2/2} + 2\int xe^{-x^2/2}\,dx + c,
\end{align}
and, by iteration: for $n \geq 4$ even
\begin{equation}
  \label{eq:gine}
  \begin{split}
  \int x^n e^{-x^2/2} \,dx
  &= -e^{-x^2/2}\bigl(x^{n-1}+(n-1)x^{n-3} + \cdots
    + (n-1) (n-3)\cdots 3 x\bigr)\\
  &\quad + (n-1) (n-3)\cdots 3 \int e^{-x^2/2} \,dx +c,
  \end{split}
\end{equation}
as well as, for $n \geq 5$ odd,
\begin{equation}
  \label{eq:gino}
  \begin{split}
  \int x^n e^{-x^2/2} \,dx
  &= -e^{-x^2/2}\bigl(x^{n-1}+(n-1)x^{n-3}+ \cdots
    + (n-1) (n-3)\cdots 4 x^2\bigr)\\
  &\quad + (n-1) (n-3)\cdots 2 \int x e^{-x^2/2} \,dx + c
  \end{split}
\end{equation}
(if $n=5$ the product $(n-1)(n-3)\cdots 4$ must be interpreted as just
equal to $4$).

Alternative expressions can be written in
terms of the incomplete Gamma function, defined as
\[
\Gamma(s,x) := \int_x^{+\infty} y^{s-1} e^{-y}\,dy
\]
for $s \in \mathbb{C}$, $\operatorname{Re} s > 1$ and $x \geq 0$ (see,
e.g., \cite{Jam:gamma}). We have to distinguish two cases: (a) if $n$
is even and $a<0$, then
\begin{align*}
\int_a^{+\infty} x^n e^{-x^2/2}\,dx 
&= \int_a^0 x^n e^{-x^2/2}\,dx + \int_0^{+\infty} x^n e^{-x^2/2}\,dx\\
&= \int_0^{-a} x^n e^{-x^2/2}\,dx + \int_0^{+\infty} x^n e^{-x^2/2}\,dx\\
&= 2\int_0^{+\infty} x^n e^{-x^2/2}\,dx - \int_{-a}^{+\infty} x^n e^{-x^2/2}\,dx\\
&= 2^{\frac{n+1}{2}} \Gamma\Bigl( \frac{n+1}{2} \Bigr)
- 2^{\frac{n-1}{2}} \Gamma\Bigl( \frac{n+1}{2},\frac{a^2}{2} \Bigr),
\end{align*}
and (b) in all other cases,
\[
  \int_a^{+\infty} x^n e^{-x^2/2}\,dx = 2^{\frac{n-1}{2}} 
  \Gamma\Bigl( \frac{n+1}{2},\frac{a^2}{2} \Bigr).
\]

We shall often use the following simple identity: for any $a \in
\erre$, $x_0, x_1 \in [-\infty,+\infty]$, and measurable $g$ such that
$x \mapsto g(x) e^{x^2/2+ax} \in L^1(x_0,x_1)$, one has
\begin{equation}
  \label{eq:inta}
  \int_{x_0}^{x_1} g(x) e^{-x^2/2+ax}\,dx =
  e^{a^2/2} \int_{x_0-a}^{x_1-a} g(x+a) e^{-x^2/2}\,dx,
\end{equation}
which follows by
$\displaystyle -\frac{x^2}{2} + ax = -\frac12(x-a)^2 + \frac12 a^2$
and a change of variable.

We conclude computing the integrals of rescaled Hermite polynomials
with respect to a standard Gaussian measure.
\begin{prop}
  \label{prop:cn}
  Let $n \in \enne$. One has
  \begin{equation}
    \label{eq:cn}
    \int_\erre h_n(\sqrt{2}x) e^{-x^2/2}\,dx
    = 2^{n/2+1/2} \Gamma(n/2+1/2).
  \end{equation}
\end{prop}
\begin{proof}
  For any $\lambda, x \in \mathbb{C}$ the generating function identity
  \[
    \exp\bigl( \lambda x - \lambda^2/2 \bigr)
    = \sum_{n=0}^\infty \frac{\lambda^n}{n!} h_n(x)
  \]
  holds, with uniform convergence of the series on compact sets (see,
  e.g., \cite[p.~7]{Mall}. Then
  \[
    e^{-\lambda^2/2} \int_\erre e^{\sqrt{2}\lambda x - x^2/2}\,dx =
    \sum_{n=0}^\infty \frac{\lambda^n}{n!} \int_\erre
    h_n(\sqrt{2}x)e^{-x^2/2} \,dx,
  \]
  where the exchange of integration and summation can be justified by
  approximation and passage to the limit. Writing
  \[
    \sqrt{2}\lambda x - \frac{x^2}{2} = -\frac12 (x-\sqrt{2}\lambda)^2
    + \lambda^2
  \]
  yields
  \[
    e^{-\lambda^2/2} \int_\erre e^{\sqrt{2}\lambda x - x^2/2} \,dx =
    e^{\lambda^2/2} \int_\erre e^{-\frac12 (x-\sqrt{2}\lambda)^2} \,dx
    = \sqrt{2\pi} \, e^{\lambda^2/2}.
  \]
  Setting
  \[
    m_n := \frac{1}{\sqrt{2\pi}} \int_\erre h_n(\sqrt{2}x)e^{-x^2/2} \,dx,
  \]
  so that
  \[
     e^{\lambda^2/2} = \sum_{n=0}^\infty \frac{\lambda^n}{n!} m_n,
  \]
  immediately implies
  \begin{equation}
    \label{eq:*}
    m_n = \left.\frac{d^n}{d\lambda^n} e^{\lambda^2/2}
    \right\vert_{\lambda=0}.
  \end{equation}
  The identity
  \begin{equation}
    \label{eq:**}
    h_n(x) = (-1)^n e^{x^2/2} \frac{d^n}{dx^n} e^{-x^2/2},
  \end{equation}
  implies, by linearity of complex differentiation,
  \[
    h_n(ix) = \frac{(-1)^n}{i^n} e^{-x^2/2} \frac{d^n}{dx^n} e^{x^2/2}
    = i^n e^{-x^2/2} \frac{d^n}{dx^n} e^{x^2/2},
  \]
  thus also
  \[
    \left. \frac{d^n}{dx^n} e^{x^2/2} \right\vert_{x=0} =
    \frac{1}{i^n} h_n(0).
  \]
  In particular, we immediately have that $m_n=0$ for every $n$ odd,
  as odd Hermite polynomials do not have terms of order zero.  We are
  going to use the following expression for the coefficients of
  Hermite polynomials (see, e.g., \cite[Eq.~18.5.13]{DLMF}):
  \[
    h_n(x) = n! \sum_{m=0}^{\lfloor n/2 \rfloor}
    \frac{(-1)^m x^{n-2m}}{m!(n-2m)! 2^m}.
  \]
  For any $n \in 2\mathbb{N}$ (the only case that matters for our
  purposes), the term of order zero has coefficient
  \[
    h_n(0) = \frac{(-1)^{n/2} n!}{2^{n/2} (n/2)!}
    = \frac{i^n n!}{2^{n/2} (n/2)!},
  \]
  hence, recalling that $\Gamma(k+1) = k!$ for any integer $k$,
  \[
    \left. \frac{d^n}{dx^n} e^{x^2/2} \right\vert_{x=0} =
    \frac{1}{i^n} h_n(0) = \frac{n!}{2^{n/2} (n/2)!} \,
    \ind{2\enne}(n) = \frac{\Gamma(n+1)}{2^{n/2} \Gamma(n/2+1)} \,
    \frac{1+(-1)^n}{2}.
  \]
  It follows by the Legendre duplication formula
  \[
  \Gamma(z) \Gamma(z+1/2) = 2^{1-2z} \, \sqrt{\pi} \, \Gamma(2z),
  \]
  taking $z=n/2+1/2$, that
  \[
  \frac{\Gamma(n+1)}{\Gamma(n/2+1)} = \frac{2^n}{\sqrt{\pi}} \Gamma(n/2+1/2),
  \]
  thus also
  \[
    \left. \frac{d^n}{dx^n} e^{x^2/2} \right\vert_{x=0} =
    \frac{2^{n/2}}{\sqrt{\pi}} \Gamma(n/2+1/2).
  \]
\end{proof}
\begin{rmk}
  As it immediately follows from equations \eqref{eq:*} and \eqref{eq:**}, the
  value on the right-hand side of equation \eqref{eq:cn} is just the absolute
  value of the term of order zero in the Hermite polynomial of order
  $n$.
\end{rmk}

\subsection{Pricing functionals}
\label{ssec:equiv}
Let $t>0$ be a fixed time and $S_t$ the discounted price at time $t$
of an asset, which is supposed to be strictly positive and, for
simplicity, with zero dividend process. The price at time zero of a
put option with strike price $\widehat{k} \geq 0$ can be written,
setting $k:=\beta_t^{-1} \widehat{k}$ and denoting the distribution
function of $\log S_t/S_0$ by $F$, as
\[
\pi(k) = \E(k-S_t)^+ = \int_\erre {(k-S_0e^x)}^+\,dF(x).
\]
Since $(k-S_0e^x)^+ = S_0{(k/S_0 - e^x)}^+$ for every $k \geq 0$ and
$x \in \erre$, we can and shall assume $S_0=1$ without loss of
generality. As is well known, the pricing functional (at time zero) of
a call option with strike $k$ on the same asset, defined by
\[
\pi_c(k) := \int_\erre {(e^x-k)}^+\,dF(x),
\]
is related to $\pi$ by the put-call parity relation
$1 - k = \pi_c - \pi$, which follows immediately by the identity
$S_t-k = (S_t-k)^+-(k-S_t)^+$.

We shall use the following properties of the function $\pi$, the short
proof of which is included for the reader's convenience. A more
detailed treatment can be found in \cite{cm:repr}.
\begin{prop}
  The functions $\pi$ is increasing, positive, $1$-Lipschitz
  continuous, and convex. Moreover, $\pi(k) \leq k$ for every
  $k \geq 0$.
\end{prop}
\begin{proof}
  Positivity and boundedness are trivial by definition, while
  monotonicity follows by $(k_1-e^x)^+ \geq (k_2-e^x)^+$ and
  $(e^x-k_1) \leq (e^x-k_2)^+$ for every $x \in \erre$ whenever
  $k_1 \geq k_2 \geq 0$. Note that $k \mapsto k-e^x$ is $1$-Lipschitz
  continuous for every $x \in \erre$. Since $y \mapsto y^+$ is
  $1$-Lipschitz continuous, so is $k \mapsto (k-e^x)^+$ by
  composition, uniformly with respect to $x \in \erre$. The property
  is then preserved integrating with respect to a measure the total
  mass of which is one. The proof of convexity is similar:
  $k \mapsto k-e^x$ is affine, in particular convex, and
  $y \mapsto y^+$ is convex increasing, hence $k \mapsto (k-e^x)^+$ is
  convex. Finally, integration with respect to a positive measure
  preserves convexity.
\end{proof}

\begin{rmk}
  All properties in the statement of the previous proposition hold
  also for call options, except for the boundedness. In fact, the
  integrand $x \mapsto (e^x-k)^+$ in the definition of $\pi_c$ is
  unbounded, which implies that $k \mapsto \pi_c(k)$ is itself
  unbounded. The boundedness of the integrand in the definition of
  $\pi$ plays a key role in the discussion to follow, and is the main
  reason for us to consider put options rather than call options.
\end{rmk}

It is clear that the distribution of logarithmic returns determines
the pricing functional for put options $\pi$. The following
proposition says that the correspondence is in fact bijective,
i.e. prices of put options for all maturities determine the
distribution of logarithmic returns. This can be seen as a non-smooth
extension of a classical result going back (at least) to Breeden and
Litzenberger \cite{BreeLitz}.
\begin{prop}
  \label{prop:Fpi}
  The map $F \mapsto \pi$ is bijective.
\end{prop}
\begin{proof}
  The map is surjective by definition. To prove injectivity, let $F_1$
  and $F_2$ be distribution functions and assume that
  \[
  \int_\erre (k-e^x)^+\,dF_1(x) = \int_\erre (k-e^x)^+\,dF_2(x),
  \]
  hence, setting $G := F_1 - F_2$ and integrating by parts,\footnote{%
  Since $F$ is not necessarily
  continuous, but just c\`adl\`ag (right-continuous with left limits),
  one has, for any c\`adl\`ag function $G$,
  \[
    F(b)G(b) - F(a)G(a) = \int_{\mathopen]a,b\mathclose]} G(x-)\,dF(x)
    + \int_{\mathopen]a,b\mathclose]} F(x)\,dG(x),
  \]
  where, if $G$ is continuous, one can obviously replace $G(x-)$ by
  $G(x)$.}
  \begin{align*}
  0 &= \int_\erre (k-e^x)^+\,dG(x) 
  = \int_{-\infty}^{\log k} (k-e^x)\,dG(x)\\
  &= \Bigl.(k-e^x) G(x)\Bigr\vert_{-\infty}^{\log k}
  + \int_{-\infty}^{\log k} e^x G(x)\,dx\\
  &= \int_{-\infty}^{\log k} e^x G(x)\,dx.
  \end{align*}
  Since this identity holds for every $k>0$, the (signed) measure with
  density $x \mapsto e^xG(x)$ with respect to the Lebesgue measure is
  equal to the zero measure, hence $G$ is equal to zero almost
  everywhere. Since $G$ is c\`adl\`ag, it follows that $G=0$
  everywhere, i.e. $F_1=F_2$.
\end{proof}

One can explicitly construct the inverse of the map $F \mapsto \pi$ as
follows: integrating by parts as in the previous proof yields
\[
  \pi(k) = \int_{-\infty}^{\log k} (k-e^x)\,dF(x) =
  \int_{-\infty}^{\log k} e^x F(x)\,dx,
\]
which implies that the c\`adl\`ag version of the derivative of $\pi$
coincides with $F(\log k)$. In particular, if $F$ is continuous, then
$\pi$ is of class $C^1$ with $\pi'(k) = F(\log k)$ for all $k>0$.

\medskip

Let us also recall that, for any fixed time to maturity, there is a
one-to-one correspondence between put option prices and the implied
volatility. Let $v_t\colon \erre_+ \to \erre_+$ be the (unique)
function satisfying $\mathsf{BS}(1,t,k,v_t(k))=\pi(k)$, where
$\mathsf{BS}(s,t,k,\sigma)$ denotes the Black-Scholes price of a put
option on an underlying with price $s$ at time zero, time to maturity
$t$, strike price $k$, interest rate equal to zero, and volatility
$\sigma$. Then we immediately have the following claim. Since here we
are concerned only with the case where the time to maturity $t$ is
fixed, we shall denote the volatility function just by $v$.
\begin{prop}
  There is a bijection between the implied volatility function $v$ and
  the distribution function $F$ of the logarithmic return.
\end{prop}

Let us assume that $F$ admits a density $f \in L^2$. As mentioned in
the introduction, we are going to construct a sequence of functions
$(f_n)$ converging to $f$ in $L^2$, hence it is natural to ask whether
the sequence of approximations $(\pi^n(k))$ defined by
\[
  \pi^n(k) := \int_\erre (k-e^x)^+\,f_n(x)\,dx, \qquad n \geq 0,
  \quad k>0,
\]
converges to $\pi(k)$ as $n \to \infty$. This is in general not the
case, because the function $x \mapsto (k-e^x)^+$ belongs to $L^\infty$
but not to $L^2$, hence it induces a continuous linear form on $L^1$,
but not on $L^2$.

One can show, however, that put option prices for all $k$ can be
reconstructed from approximation to option prices with payoff of the
type
\[
  \theta_{k_1,k_2}(x) = (k_2-e^x)^+ - \frac{k_2}{k_1}(k_1-e^x)^+,
  \qquad k_1, k_2 > 0.
\]
More precisely, to identify the pricing functional $\pi$, it suffices
to know, for any sequence $(f_n)$ converging to $f$ in $L^2(\erre)$,
the values $\ip{\theta_{k_1,k_2}}{f_n}$ for all $k_1,k_2>0$ and all
$n \geq 0$, where we recall that $\ip{\cdot}{\cdot}$ stands for the
scalar product of $L^2$.
In fact, since $\theta_{k_1,k_2} \in L^2$, for any sequence $(f_n)$
converging to $f$ in $L^2$ one has
\[
  \pi^n(k_2) - \frac{k_2}{k_1} \pi^n(k_1) =
  \ip[\big]{\theta_{k_1,k_2}}{f_n} \longrightarrow
  \ip[\big]{\theta_{k_1,k_2}}{f} = \pi(k_2) - \frac{k_2}{k_1} \pi(k_1).
\]
Moreover, the function $x \mapsto \frac{k_2}{k_1} (k_1 - e^x)^+$
converges to zero as $k_1 \to 0$ in $L^p$ for every
$p \in [1,\infty\mathclose[$, hence
\begin{equation}
  \label{eq:coa}
  \lim_{k_1 \to 0} \frac{k_2}{k_1} \pi(k_1) =
  \lim_{k_1 \to 0} \int_\erre \frac{k_2}{k_1} (k_1-e^x)^+ f(x)\,dx = 0,
\end{equation}
i.e.
\[
  \lim_{k_1 \to 0} \lim_{n \to \infty} \ip[\big]{\theta_{k_1,k_2}}{f_n}
  = \pi(k_2) \qquad \forall k_2 > 0
\]
(see \cite{cm:repr} for more detail). Taking into account
Proposition~\ref{prop:Fpi}, the proof of the following claim is then
immediate.
\begin{prop}
  If there exists a sequence $(f_n) \subset L^2$ converging to $f$ in
  $L^2$, then there is a bijection between
  \[
    \bigl\{ \ip{\theta_{k_1,k_2}}{f_n} : \,
    k_1, k_2 \in \mathopen]0,\infty\mathclose[, \, n \in \mathbb{N}
    \bigr\}
  \]
  and the distribution $F$ of logarithmic returns.
\end{prop}
\noindent Completely analogously, if $\pi(k_1)$ is known, then
\[
  \pi(k_2) = \frac{k_2}{k_1} \pi(k_1)
  + \lim_{n \to \infty} \ip[\big]{\theta_{k_1,k_2}}{f_n}
  = \frac{k_2}{k_1} \pi(k_1)
  + \lim_{n \to \infty} \Bigl( \pi^n(k_2)
  - \frac{k_2}{k_1}\pi^n(k_1) \Bigr).
\]
Note that although the function
$x \mapsto \frac{k_2}{k_1} (k_1 - e^x)^+$ does not converge to zero in
$L^\infty$ as $k_1 \to 0$, because
\[
  \sup_{x \in \erre} \frac{k_2}{k_1} (k_1 - e^x)^+ = k_2,
\]
the convergence in equation \eqref{eq:coa} also holds with $f \in L^1$,
i.e. without any extra integrability assumption on $f$, because
$\frac{k_2}{k_1} (k_1 - e^x)^+ f(x) \leq k_2 f(x)$ for every
$x \in \erre$, hence the claim follows by dominated convergence.


\section{Pricing estimates via Hermite series expansion}
\label{sec:Hp}
We are going to discuss the construction and some properties of a
class of approximations of the pricing functional $\pi$ for put
options with fixed time to maturity based on Hermite series expansion
of the density of logarithmic returns. Particular attention is given
to reducing as many computations as possible to integrals of
polynomials with respect to Gaussian measures. This is desirable in
practical implementations because such integrals, as seen in
\S\ref{ssec:gpi}, can be numerically computed in an efficient way.

Recall that time to maturity, denoted by $t$, is fixed. Let us define
the ($\cF_t$-measurable) random variable $X$ by
\[
S_t e^{\overline{q}t} = S_0 \exp\bigl( \sigma X + m \bigr),
\]
where $m$ and $\sigma>0$ are constants, and
\[
\overline{q} := \frac{1}{t} \int_0^t q_s\,ds
\]
is the mean dividend rate over the time interval $[0,t]$.
We assume that the law of $X$ admits a density $f$. Then
\begin{align*}
  \pi(k) = \E{(k-S_t)}^+ 
  &= \E \bigl( k - S_0 e^{\sigma X + m - \overline{q}t} \bigr)^+\\
  &= \int_\erre \bigl( k - S_0 e^{\sigma x + m - \overline{q}t} \bigr)^+ f(x)\,dx.
\end{align*}
Moreover, as the discounted yield process associated to the
(discounted) price process $S$ is a martingale, one has
\[
\int_\erre e^{\sigma x} f(x)\,dx = e^{-m}.
\]
Let us further assume that the density $f$ belongs to $L^2$, i.e. that
\[
\int_\erre f(x)^2\,dx < \infty.
\]
Note that $f \in L^1$ by definition of density, hence $f$ is
automatically in $L^2$ if, for instance, it is bounded (which is often
the case for many parametric families of densities that are used to
model returns).

By Lemma~\ref{lm:sigma} there exists $(\alpha_n) \in \ell^2$ such that
the sequence of functions $(f_N)$ defined by
\[
  f_N(x) := \sum_{n=0}^N \alpha_n h_n(\sqrt{2}x) e^{-x^2/2}, \qquad N
  \geq 0,
\]
converges to $f$ in $L^2$.
Setting
\[
\zeta_+ := \frac{1}{\sigma} \Bigl(
\log \frac{k}{S_0} - m + \overline{q}t \Bigr),
\]
one has
\[
\E{(k-S_t)}^+ = \int_{-\infty}^{\zeta_+} \bigl( k - e^{\sigma x +
m - \overline{q}t} \bigr) f(x)\,dx.
\]
Replacing $f$ by $f_N$ in the previous formula, one obtains
\begin{align*}
  \pi^N &:= \int_{-\infty}^{\zeta_+} \bigl( k - S_0 e^{\sigma x +
    m - \overline{q}t} \bigr) f_N(x)\,dx\\
  &= k \int_{-\infty}^{\zeta_+} f_N(x)\,dx - e^{-\overline{q}t}
  S_0\int_{-\infty}^{\zeta_+} f_N(x) e^{\sigma x + m} \,dx.
\end{align*}
Setting $\overline{f}_N(x):= e^{x^2/2} f_N(x) = \sum_{n=0}^N \alpha_n
h_n(\sqrt{2}x)$ and
\[
\zeta_- := \zeta_+ - \sigma = \frac{1}{\sigma} \Bigl( \log
\frac{k}{S_0} - m - \sigma^2 + \overline{q}t \Bigr),
\]
and writing
\[
-\frac12 x^2 + \sigma x + m = -\frac12 (x-\sigma)^2 + \frac12\sigma^2 + m,
\]
we have
\begin{align*}
  \int_{-\infty}^{\zeta_+} f_N(x) e^{\sigma x + m} \,dx
  &= e^{\sigma^2/2+m} \int_{-\infty}^{\zeta_+} \overline{f}_N(x)
     e^{-(x-\sigma)^2/2}\,dx\\
  &= e^{\sigma^2/2+m} \int_{-\infty}^{\zeta_-} \overline{f}_N(x+\sigma) 
     e^{-x^2/2}\,dx.
\end{align*}
Therefore
\[
\pi^N = k \int_{-\infty}^{\zeta_+} f_N(x)\,dx
  - e^{\sigma^2/2+m-\overline{q}t} S_0 \int_{-\infty}^{\zeta_-} 
\overline{f}_N(x+\sigma) e^{-x^2/2}\,dx,
\]
where
\begin{gather*}
  \int_{-\infty}^{\zeta_+} f_N(x)\,dx = \sum_{n=0}^N
  \alpha_n \int_{-\infty}^{\zeta_+} h_n(\sqrt{2}x) e^{-x^2/2}\,dx,\\
  \int_{-\infty}^{\zeta_-} \overline{f}_N(x+\sigma)
  e^{-x^2/2}\,dx = \sum_{n=0}^N \alpha_n \int_{-\infty}^{\zeta_-}
  h_n\bigl(\sqrt{2}(x+\sigma)\bigr) e^{-x^2/2}\,dx.
\end{gather*}
Note that all integrals with respect to the Gaussian density appearing
in the above expansions can be computed in closed form, in terms of
the Gaussian density and distribution functions, or in terms of
incomplete Gamma functions, as shown in \S\ref{ssec:gpi}.
\begin{rmk}
  The Black-Scholes formula is a special case of the above with $N=0$,
  replacing $\sigma$ and $m$ by $\sigma_0\sqrt{t}$ and
  $-\frac12 \sigma_0^2 t$, respectively, where $\sigma_0$ stands for
  the volatility of the underlying.
\end{rmk}

\medskip

We now discuss some properties of this class of approximations:
\begin{itemize}
\item[(i)] As it follows by \S\ref{ssec:equiv}, the convergence of
  $f_N$ to $f$ in $L^2$ as $N \to \infty$ does not imply that
  $\pi^N \to \pi$, but one has nonetheless enough information to
  uniquely determine $\pi$.
\item[(ii)] The function $f_N$ in general is not a density, as it is
  not guaranteed to be positive and its integral over the real line is
  not necessarily equal to one. Furthermore, in general $f_N$ does not
  converge to $f$ in $L^1$. It is known, however, that if $f \in L^p$,
  with $p \in \mathopen]4/3,4\mathclose[$, then $f_N \to f$ in $L^p$
  (see~\cite{AsWa}, where the authors prove that the result is sharp,
  in the sense that convergence fails for $p \in [1,4/3]$ and for
  $p \geq 4$, and \cite{Mucken:HL}). This is the case, for instance,
  if the density $f$ is bounded, in which case, by interpolation
  between $L^1$ and $L^\infty$, $f$ belongs to $L^p$ for every
  $p \in [1,\infty]$.
\item[(iii)] The martingale condition
  \[
    \int_\erre e^{\sigma x} f(x)\,dx = e^{-m}
  \]
  is not preserved substituting $f$ with $f_N$. However, a kind of
  ``asymptotic martingale property'' holds: note that
  \[
    \lim_{a \to \infty} \int_{-\infty}^a e^{\sigma x} f(x)\,dx =
    e^{-m}
  \]
  from below. Let $\varepsilon>0$ be arbitrary but fixed. Then there
  exists $a_0=a_0(\varepsilon)$ such that for every $a>a_0$
  \[
    e^{-m} - \varepsilon/2 \leq \int_{-\infty}^a e^{\sigma x} f(x)\,dx
    \leq e^{-m}.
  \]
  Let $a > a_0$ be arbitrary but fixed. By the Cauchy-Schwarz
  inequality,
  \begin{align*}
  \int_{-\infty}^a e^{\sigma x} \abs[\big]{f_N(x)-f(x)}\,dx
  &\leq \biggl( \int_{-\infty}^a e^{2\sigma x}\,dx\biggr)^{1/2}
    \biggl( \int_{-\infty}^a \abs[\big]{f_N(x)-f(x)}^2 \,dx\biggr)^{1/2}\\
  &\leq \biggl( \frac{e^{2\sigma a} - 1}{2\sigma} \biggr)^{1/2}
    \norm[\big]{f_N - f}_{L^2},
  \end{align*}
  hence
  \begin{equation}
    \label{eq:conva}
    \lim_{N \to \infty} \int_{-\infty}^a e^{\sigma x} f_N(x)\,dx
    = \int_{-\infty}^a e^{\sigma x} f(x)\,dx,
  \end{equation}
  i.e. there exists $N_0=N_0(a,\varepsilon)$ such that, for every
  $N > N_0$,
  \[
    \int_{-\infty}^a e^{\sigma x} f(x)\,dx - \varepsilon/2 \leq
    \int_{-\infty}^a e^{\sigma x} f_N(x)\,dx \leq
    \int_{-\infty}^a e^{\sigma x} f(x)\,dx + \varepsilon/2,
  \]
  hence
  \[
    e^{-m} - \varepsilon \leq \int_{-\infty}^a e^{\sigma x} f_N(x)\,dx
    \leq e^{-m} + \varepsilon/2.
  \]
\end{itemize}
Some of the above issues can be avoided assuming that there exists
$\delta>0$ such that
\[
  \widetilde{f} \colon x \mapsto e^{\sigma(1+\delta)\abs{x}}f(x) \in L^2.
\]
In fact, let $(\widetilde{f}_N)$ be a sequence of function converging
to $\widetilde{f}$ in $L^2$ and define $(f_N)$ by
\[
  e^{\sigma(1+\delta)\abs{x}} f_N(x) = \widetilde{f}_N
  \qquad \forall N \geq 0.
\]
Then
\begin{align*}
  \int_\erre \abs[\big]{f_N(x)-f(x)}\,dx 
  &= \int_\erre e^{-\sigma(1+\delta)\abs{x}} e^{\sigma(1+\delta)\abs{x}}
    \abs[\big]{f_N(x)-f(x)}\,dx\\
  &\leq \biggl( \int_\erre e^{-2\sigma(1+\delta)\abs{x}}\,dx \biggr)^{1/2}
    \biggl( \int_\erre e^{2\sigma(1+\delta)\abs{x}}
    \abs[\big]{f_N(x)-f(x)}^2\,dx \biggr)^{1/2},
\end{align*}
hence $f_N \to f$ in $L^1(\erre)$ as $N \to \infty$. In particular,
even though $f_N$ is in general not a density, as its $L^1$ norm may
not be equal to one, it does converge to a density as $N \to \infty$,
in the sense that ${\norm{f_N}}_{L^1} \to 1$. Moreover, as discussed
in subsection \ref{ssec:equiv}, convergence of $f_N$ to $f$ in $L^1$ implies
convergence of put option prices, i.e.  $\pi^N(k) \to \pi(k)$ for
every $k>0$. We also have
\[
  \int_\erre e^{\sigma x} \abs[\big]{f_N(x)-f(x)}\,dx
  = \int_\erre e^{\sigma x} e^{-\sigma(1+\delta)\abs{x}} e^{\sigma(1+\delta)\abs{x}}
    \abs[\big]{f_N(x)-f(x)}\,dx,
\]
where
$e^{\sigma x} e^{-\sigma(1+\delta)\abs{x}} \leq
e^{-\sigma\delta\abs{x}}$ for all $x \in \erre$, hence, since
$x \mapsto e^{-\sigma\delta\abs{x}} \in L^2$,
\[
  \lim_{N \to +\infty} \int_\erre e^{\sigma x} \abs[\big]{f_N(x)-f(x)}\,dx
  \lesssim \lim_{N \to +\infty}
  \biggl( \int_\erre e^{2\sigma(1+\delta)\abs{x}}
    \abs[\big]{f_N(x)-f(x)}^2\,dx \biggr)^{1/2} = 0.
\]
In particular,
\[
  \lim_{N \to +\infty} \int_\erre e^{\sigma x} f_N(x)\,dx
  = \int_\erre e^{\sigma x} f(x)\,dx,
\]
which is a kind of asymptotic martingale property improving upon equation
\eqref{eq:conva}.
Approximate pricing formulas for put options involving only integrals
of polynomials with respect to a standard Gaussian measure can also be
obtained proceeding analogously to the case treated above, even though
computations are more cumbersome. For the sake of completeness, full
detail is provided in the appendix.

The extra integrability assumption, however, could be too strong for
certain applications, as it implies that $X$ admits exponential
moments. In fact, if $x \mapsto e^{\alpha \abs{x}} f(x) \in L^2$,
then, for any $\beta<\alpha/2$, the Cauchy-Schwartz inequality yields
\begin{align*}
  \E e^{\beta \abs{X}} = \int_\erre e^{\beta \abs{x}} f(x) \,dx
  &= \int_\erre e^{(\beta-\alpha/2) \abs{x}} e^{\alpha/2 \abs{x}} f(x) \,dx\\
  &\leq \biggl( \int_\erre e^{(2\beta-\alpha) \abs{x}} \,dx \biggr)^{1/2}
  \biggl( \int_\erre e^{\alpha \abs{x}} f(x)^2 \,dx \biggr)^{1/2} < \infty.
\end{align*}
Note that exponential integrability of the return $X=\log S_T$ is not
needed to ensure that $S_T$ has finite expectation.


\section{Calibration of approximate pricing functionals}
\label{sec:cal}
For any $m \in \erre$, $\sigma \in \erre_+$, and
$\alpha = (\alpha_0,\ldots,\alpha_N) \in \erre^{N+1}$, the approximate
pricing method introduced in Section \ref{sec:Hp} can be represented as a
function $k \mapsto \widehat{\pi}(k;m,\sigma,\alpha)$, where
$m,\sigma,\alpha$ are treated as parameters (we omit the variable $t$
because we assume, as before, that time to maturity is fixed). Let
$(k_i)_{i \in I}$ be a set of strike prices for which prices of put
options $(\pi_i)_{i \in I}=(\pi(k_i))_{i \in I}$ are
observed. Moreover, we assume that $f_N \to f$ in $L^1$, so that the
correction procedure described in subsection \ref{ssec:equiv} is not
necessary. Even though this is a loss of (theoretical) generality, it
does \emph{not} imply any loss of precision in the empirical analysis
carried out in the next section.

The approximating Hermite pricing model with parameters
$(m,\sigma,\alpha)$ can be calibrated to observed prices via a
minimization problem of the form
\[
  \inf_{(m,\sigma,\alpha) \in \Theta} J(m,\sigma,\alpha),
  \qquad J(m,\sigma,\alpha) := L\bigl( (\pi_i),
    (\widehat{\pi}(k_i;m,\sigma,\alpha)) \bigr),
\]
where $\Theta$ stands for a subset of
$\erre \times \erre_+ \times \erre^{N+1}$ and $L$ is a loss function
defined on $\ell(I) \times \ell(I)$, with $\ell(I)$ denoting the
vector space of sequences indexed by the set $I$.
Since our main interest is the minimization of the relative pricing
error, we shall set
\[
  L(x,y) := \norm[\Big]{\frac{y-x}{x}}
  = \norm[\Big]{\frac{y}{x} - 1},
\]
where $y/x$ is defined pointwise, i.e. $(y/x)_i:=y_i/x_i$ for every
$i \in I$, and $\norm{\cdot}$ is a norm on $\ell(I)$, typically the
$\ell^2$ norm, corresponding to ordinary least squares, or the
$\ell^1$ norm, corresponding to least absolute deviation. Note that
$L(x,y)=+\infty$ as soon as $x_i=0$ for some $i \in I$. However, in
practice this does not cause trouble because no options with price
zero are traded anyway. On the other hand, out-of-the-money options
with very short time to maturity will have prices close to zero, hence
calibration is sensitive to the presence of such option prices in the
set $(\pi_i)_{i \in I}$. In practice, this is also not too
problematic, as one could use weighted norms on $\ell(I)$, or just
disregard options with prices too close to zero, i.e. select a
suitable subset $I'$ of the index set $I$.

Let us write the objective function $J$ as $J=\norm{R}$, with
\[
  R(m,\sigma,\alpha) = \biggl(\frac{1}{\pi_i}
    \int_\erre {(k_i-e^{\sigma x+m})}^+
    f_\alpha(x) e^{-\sigma(1+\delta)\abs{x}}\,dx - 1 \biggr)_{i \in I},
\]
where
\[
f_\alpha(x) := \sum_{j=0}^N \alpha_j h_j(\sqrt{2}x) e^{-x^2/2}.
\]
Denoting the cardinality of $I$ by $\abs{I}$, the relative error $R$
can be seen as a function from (a subset of)
$E := \erre \times \mathopen]0,\infty\mathclose[ \times \erre^{1+N}$
to $\erre^{\abs{I}}$, which turns out to be very regular.
\begin{prop}
  Let $n:=\abs{I}$. The relative error function $R$ belongs to
  $C^\infty(E;\erre^n)$.
\end{prop}
\begin{proof}
  Let the function
  $\zeta \colon \mathopen]0,\infty\mathclose[^n \times \erre \times
  \mathopen]0,\infty\mathclose[ \to \erre^n$ be defined by
  $\zeta(k,m,\sigma) = \frac{1}{\sigma} (\log k - m)$, with the
  logarithm taken componentwise. Then
  \begin{align*}
    \int_\erre {(k-e^{\sigma x + m})}^+ f_\alpha(x) e^{-\sigma(1+\delta)\abs{x}}\,dx
    &= \int_{-\infty}^{\zeta} (k-e^{\sigma x + m}) f_\alpha(x)
      e^{-\sigma(1+\delta)\abs{x}} \,dx\\
    &:= \biggl( \int_{-\infty}^{\zeta_i} (k_i-e^{\sigma x + m}) f_\alpha(x)
      e^{-\sigma(1+\delta)\abs{x}} \,dx \biggr)_{i=1,\ldots,n}
  \end{align*}
  The function $\alpha \mapsto f_\alpha(x)$ is linear, hence of class
  $C^\infty$ for every $x \in \erre$.
  Moreover, the functions $(m,\sigma) \mapsto e^{\sigma x+m}$ and
  $\sigma \mapsto e^{-\sigma(1+\delta)\abs{x}}$ are also of class
  $C^\infty$ for every $x \in \erre$. It follows immediately that
  $(m,\sigma,\alpha) \mapsto g(x;m,\sigma,\alpha) := (k-e^{\sigma x +
    m}) f_\alpha(x) e^{-\sigma(1+\delta)\abs{x}}$ is of class
  $C^\infty$ for every $x \in \erre$. Elementary calculus shows that
  derivatives of any order of
  $(m,\sigma,\alpha) \mapsto g(\cdot;m,\sigma,\alpha)$ are integrable
  on $\mathopen]-\infty,\zeta_i]$ for every $i=1,\ldots,n$, and
  $\zeta$ it itself of class $C^\infty$. Noting that
  $k-e^{\sigma\zeta+m}=0$ by definition of $\zeta$, the claim follows
  by the Leibniz rule for differentiation under the integral sign.
\end{proof}
The derivatives of $R$ can be computed easily: assuming for simplicity
$\abs{I}=1$, one has
\begin{align*}
  \partial_m R(m,\sigma,\alpha)
  &= \int_{-\infty}^\zeta e^{\sigma x + m} e^{-\sigma(1+\delta)\abs{x}}
    f_\alpha(x)\,dx\\
  \partial_\sigma R(m,\sigma,\alpha)
  &= -\int_\erre {(k-e^{\sigma x + m})}^+ e^{-\sigma(1+\delta)\abs{x}}
    (1+\delta)\abs{x} f_\alpha(x)\,dx\\
  &\quad - \int_{-\infty}^\zeta e^{\sigma x + m} e^{-\sigma(1+\delta)\abs{x}}
    x f_\alpha(x)\,dx,\\
  \partial_{\alpha_j} R(m,\sigma,\alpha)
  &= \int_\erre {(k-e^{\sigma x+m})}^+
    e^{-\sigma(1+\delta)\abs{x}} h_j(\sqrt{2}x) e^{-x^2/2}\,dx.
\end{align*}
Explicit expressions can also be obtained for derivatives of higher
order, which can be useful to check numerically first and second-order
conditions for optimality. For instance, if the norm in the definition
of $J$ is the $\ell^2$ norm, then the function
$(m,\sigma,\alpha) \mapsto \norm{R(m,\sigma,\alpha)}^2$ is
continuously differentiable and its (Fr\'echet) derivative is
$2\ip{R}{R'}$, where $R'$ can be identified with the $n$
$\erre^{N+3}$-valued functions
\[
  \bigl( \partial_m R_i, \partial_\sigma R_i, \partial_{\alpha_0},\ldots,
  \partial_{\alpha_N} R_i \bigr), \qquad i=1,\ldots,n.
\]
On the other hand, using the above explicit expressions to identify
possible local minima solving $\ip{R}{R'}=0$ may not be feasible, as
the equation is highly nonlinear.

If $J = \norm{R}_{\ell^1}$, that is, if the optimality criterion is
defined in terms of least absolute deviation, then $J$ is not
differentiable, because the $\ell^1$ norm is not. For practical
purposes, this suggests that derivative-free minimization algorithms
should be preferred.

\medskip

We are now going to discuss a convexity properties of $J$ with respect
to the variable $\alpha$ for fixed $m$ and $\sigma$.
It follows immediately from Section \ref{sec:Hp} that it is possible to
write
\[
\widehat{\pi}(k_i;m,\sigma,\alpha) = \sum_{j=0}^N \bigl( 
k_i \Phi^1_j(k_i;m,\sigma) - S_0 \Phi^2_j(k_i;m,\sigma) \bigr) \alpha_j,
\]
where $\Phi^1$ and $\Phi^2$ are $\erre^{N+1}$-valued functions
depending on the parameters $m$ and $\sigma$, but not on $\alpha$.
Therefore, defining the matrix $\Psi \in \erre^{n \times (N+1)}$ by
\begin{equation}
  \label{eq:Psi}
  \Psi_{ij} := \frac{1}{\pi_i} \bigl( k_i \Phi^1_j - S_0\Phi^2_j \bigr),
\end{equation}
we have
\[
  J(m,\sigma,\alpha) = \norm[\big]{\Psi(m,\sigma)\alpha-1}.
\]
Although the objective function $J$ is not convex, the function
$\alpha \mapsto J(m,\sigma,\alpha)$ is convex. This observation is
useful in view of the identity
\[
  \inf_{m,\sigma,\alpha} J(m,\sigma,\alpha)
  = \inf_{m,\sigma} \inf_\alpha J(m,\sigma,\alpha),
\]
where the minimizers of $\alpha \mapsto \norm{\Psi\alpha-1}$ can be
characterized by $\partial \norm{\Psi\alpha-1}=0$, with $\partial$
denoting the subdifferential in the sense of convex analysis. If
$\norm{\cdot}$ is the $\ell^2$ norm, then the function
$\alpha \mapsto \norm{\Psi\alpha-1}^2$ is Fr\'echet differentiable
with derivative $v \mapsto 2\ip{\Psi\alpha-1}{\Psi v}_{\ell^2}$, hence a
minimizer $\alpha_*=\alpha_*(m,\sigma)$ is characterized by
$\Psi^\top(\Psi\alpha_* - 1)=0$. In particular, if $\Psi^\top\Psi$ is
invertible, then the minimizer is unique and equal to
\[
\alpha_* = (\Psi^\top\Psi)^{-1}\Psi^\top 1_n,
\]
where $1_n=(1,\ldots,1) \in \erre^n$. Of course $\alpha_*$ is nothing
else than the estimate of $\alpha$ by ordinary least squares.

If instead the $\ell^1$ norm is used in the definition of $J$, the
function $\alpha \mapsto \norm{\Psi\alpha - 1}$ is not differentiable,
and its subdifferential is multivalued, hence not easy to deal with.
However, the minimization problem
$\inf_\alpha \norm{\Psi\alpha - 1}_{\ell^1}$ can be solved by linear
programming, writing it in the equivalent form
\begin{align*}
  \inf_{u,\alpha}\, &{\ip{1_n}{u}}_{\erre^n}\\
  &\text{s.t. } u \geq \Psi\alpha-1,\\
  &\phantom{\text{s.t. }} u \geq -(\Psi\alpha-1),
\end{align*}
or equivalently, in coordinates,
\begin{align*}
  \inf_{u,\alpha}\, &\sum_{i=1}^n u_i\\
  &\text{s.t. }  u_i \geq (\Psi\alpha)_i - 1,\\
  &\phantom{\text{s.t. }} u_i \geq -(\Psi\alpha)_i + 1
  \quad \forall i=1,\ldots,n.
\end{align*}
As already mentioned, using the $\ell^1$ norm is equivalent to
estimating $\alpha$ by least absolute deviation, a method that is less
sensitive to outliers than ordinary least squares, which corresponds
to using the $\ell^2$ norm.

\medskip

We are now going to consider additional constraints on $\alpha$, for
fixed $m$ and $\sigma$, implying that the approximation $f_N$ to the
density $f$ integrates to one and satisfies an approximate martingale
condition, i.e. that
\begin{equation}
  \label{eq:vincoli}
  \int_\erre f_N(x)\,dx = 1 \quad \text{ and } \quad
  \int_\erre e^{\sigma x + m} f_N(x)\,dx = 1,
\end{equation}
respectively. Defining the vector $c=(c_0,c_1,\ldots,c_N) \in \erre^{1+N}$
by
\[
c_n := \int_\erre h_n(\sqrt{2}x) e^{-x^2/2} \,dx,
\]
the first condition in equation \eqref{eq:vincoli} can be written as
\[
  \ip[\big]{c}{\alpha}_{\erre^{1+N}} = c_0\alpha_0 + c_1\alpha_1 +
  \cdots + c_N \alpha_N = 1.
\]
The vector $c$ can be computed in close form thanks to
Proposition~\ref{prop:cn}.

The approximate martingale condition, that is the second condition in
equation \eqref{eq:vincoli}, is equivalent to
\begin{equation}
  \label{eq:AMP}
  \sum_{n=0}^N \alpha_n \int_\erre h_n(\sqrt{2}x)
  e^{\sigma x - x^2/2}\,dx = e^{-m},
\end{equation}
where, by equation \eqref{eq:inta},
\[
  \int_\erre h_n(\sqrt{2}x) e^{\sigma x - x^2/2}\,dx
  = e^{\sigma^2/2} \int_\erre h_n(\sqrt{2}(x+\sigma)) e^{-x^2/2}\,dx.
\]
We are going to obtain closed-form expressions for the coefficients of
the polynomial $F_n \in \erre[\sigma]$ defined by
\[
F_n(\sigma) := \int_\erre h_n(\sqrt{2}(x+\sigma)) e^{-x^2/2}\,dx.
\]
More generally, let $P_n(x) \in \erre[x]$ be a polynomial of degree
$n$, and let us compute the coefficient of the polynomial in
$\erre[\sigma]$ defined by
\[
\int_\erre P_n(x+\sigma) e^{-x^2/2}\,dx.
\]
Writing $P_n(x) = a_0 + a_1x + \cdots + a_n x^n$, it is clear that
let us first compute, for any $m \in \enne$,
\[
\int_\erre (x+\sigma)^m e^{-x^2/2}\,dx.
\]
One has
\[
(x+\sigma)^m = \sum_{k=0}^m \binom{m}{k} x^{m-k} \sigma^k,
\]
hence
\[
  \int_\erre (x+\sigma)^m e^{-x^2/2}\,dx = \sum_{k=0}^m \sigma^k
  \binom{m}{k} \int_\erre x^{m-k} e^{-x^2/2}\,dx.
\]
It follows by the definition of the gamma function that
\[
  g(m,k) := \binom{m}{k} \int_\erre x^{m-k} e^{-x^2/2}\,dx
  = \binom{m}{k} 2^{\frac{m-k-1}{2}} \bigl( 1 + (-1)^{m-k} \bigr)
  \Gamma\Bigl( \frac{m-k+1}{2} \Bigr),
\]
hence
\[
  \int_\erre (x+\sigma)^m e^{-x^2/2}\,dx = \sum_{k=0}^m \sigma^k
  g(m,k),
\]
and finally
\begin{align*}
  F_n(\sigma) := \int_\erre P_n(x+\sigma) e^{-x^2/2}\,dx
  &= \sum_{m=0}^n a_m \int_\erre (x+\sigma)^m e^{-x^2/2}\,dx\\
  &= \sum_{m=0}^n a_m \sum_{k=0}^m g(m,k) \sigma^k\\
  &= \sum_{k=0}^n \biggl( \sum_{m \geq k}^n a_m g(m,k) \biggr)\sigma^k.
\end{align*}
Choosing $P_n(x) := h_n(\sqrt{2}x)$, the approximate martingale
condition \eqref{eq:AMP} can be written as
\[
  \alpha_0 F_0(\sigma) + \alpha_1 F_1(\sigma) + \cdots + \alpha_N
  F_N(\sigma) = e^{-m-\sigma^2/2}.
\]
The first few polynomials $F_n(\sigma)$ are
\begin{align*}
  F_0(\sigma) &= \sqrt{2\pi}, & F_1(\sigma) &= 2 \sqrt{\pi} \sigma,\\
  F_2(\sigma) &= \sqrt{2\pi} + 2\sqrt{2\pi} \sigma^2, &
  F_3(\sigma) &= 6\sqrt{\pi}\sigma + 4 \sqrt{\pi}\sigma^3\\
  F_4(\sigma) &= 3\sqrt{2\pi} + 12\sqrt{2\pi}\sigma^2
                 + 4\sqrt{2\pi} \sigma^4, &
  F_5(\sigma) &= 30\sqrt{\pi} \sigma + 40\sqrt{\pi} \sigma^3
                 + 8\sqrt{\pi} \sigma^5.
\end{align*}
Both constraints in equation \eqref{eq:vincoli} are affine in $\alpha$ (for
fixed $m$ and $\sigma$), in particular they are convex, as well as
their intersection $\Theta = \Theta(m,\sigma) \subset
\erre^{1+N}$. Moreover, since $F_n(\sigma)>0$ for every $\sigma>0$ and
$n \in \enne$, and $c_n=0$ for every odd $n$, the two constraints are
non-redundant. The minimization with the constraints in equation
\eqref{eq:vincoli} then becomes
\begin{equation}
  \label{eq:mimo}
\inf_{\alpha \in \Theta} \norm[\big]{\Psi(m,\sigma)\alpha - 1},
\end{equation}
which is still a convex minimization problem. If the norm is the
$\ell^1$ norm, adding the constraint $\alpha \in \Theta$ to the linear
programming formulation of the minimization is trivial. On the other
hand, if the norm is the $\ell^2$ norm, we can no longer use ordinary
least squares, but the minimization problem can be solved by quadratic
programming. In fact, setting $n:=\abs{I}$,
\[
  \norm[\big]{\Psi\alpha - 1}_{\ell^2}^2 =
  \ip{\Psi^\top\Psi\alpha}{\alpha} - 2 \ip{\Psi^\top 1_n}{\alpha} + n
\]
hence the minimization of $\norm{\Psi\alpha - 1}_{\ell^2}$ over
$\Theta$ is equivalent to the quadratic programming problem
\begin{align*}
  \inf_\alpha \, & \bigl( \ip{\Psi^\top\Psi\alpha}{\alpha} - 2 \ip{\Psi^\top 1_n}
                   {\alpha} \bigr)\\
  &\text{s.t. } \ip{c}{\alpha} = 1,\\
  &\phantom{\text{s.t. }} \ip{F(\sigma)}{\alpha} = e^{-m-\sigma^2/2}.
\end{align*}

\begin{rmk}
  If $N=1$, the constrained minimization problem \eqref{eq:mimo}
  degenerates, in the sense that the constraints already uniquely
  identify the solution. In fact, the two affine equations
  $\ip{c}{\alpha}=1$ and
  $\ip{F(\sigma)}{\alpha} = \exp(-m-\sigma^2/2)$ have a unique
  solution because the vectors $c$ and $F(\sigma)$ are
  independent. Similarly, if $N=2$, each constraint identifies a plane
  of $\erre^3$, hence their intersection is a line in $\erre^3$,
  i.e. the constrained minimization problem can be reduced, by a
  reparametrization, to an unconstrained minimization problem in one
  real variable. More precisely, let $A$ be the matrix defined by
  \[
  A =
  \begin{bmatrix}
    F \\ c
  \end{bmatrix},
  \]
  with $c$ and $F$ considered as row vectors, $v$ a vector in
  $\erre^3$ generating the kernel of $A$, and $\alpha_0$ any vector in
  $\Theta$, i.e. any solution to the equation
  \begin{equation}
    \label{eq:mater}
    \begin{bmatrix}
      F \\ c
    \end{bmatrix}
    \alpha_0 =
    \begin{bmatrix}
      e^{-m-\sigma^2/2} \\ 1
    \end{bmatrix}.
  \end{equation}
  Then $\Theta = \{\alpha_0 + av\}_{a \in \erre}$. Recalling that
  \[
    A = \begin{bmatrix}
      F \\ c
    \end{bmatrix}
    = \sqrt{2\pi}
    \begin{bmatrix}
      1 & 2\sigma & 1 + \sigma^2 \\
      1 & 0 & 1
    \end{bmatrix},
  \]
  explicit computations show that a generator of the kernel of $A$ is
  $(1,\sigma/2,-1)$, and a solution to equation \eqref{eq:mater} is
  \[
    \alpha_0 = \frac{1}{\sqrt{2\pi}}
    \biggl(1, \frac{e^{-m-\sigma^2/2} - 1}{2\sigma}, 0 \biggr).
  \]
\end{rmk}

\medskip

Finally, assume that, for given $m$ and $\sigma$,
$\alpha_* = \alpha_*(m,\sigma)$ is a minimizer of the function
$\alpha \mapsto \norm{\Psi\alpha - 1}$, with or without the
constraints in equation \eqref{eq:vincoli}, and recall that $\Psi$ depends on
$m$ and $\sigma$, but not on $\alpha$. Then
\begin{align*}
  \inf_{m,\sigma,\alpha} J(m,\sigma,\alpha)
  &= \inf_{m,\sigma} \norm[\big]{\Psi(\sigma,m)\alpha_*(\sigma,m)-1}\\
  &= \inf_{m,\sigma} \norm[\bigg]{\biggl(\frac{1}{\pi_i}
    \int_\erre {(k-e^{\sigma x+m})}^+
    f_{\alpha_*(m,\sigma)}(x)\,dx - 1 \biggr)}.
\end{align*}
Unfortunately it does not look possible to make any claim about the
convexity of the function to be minimized. Therefore, results obtained
by numerical minimization may depend on the initialization and may get
trapped at local minima. Empirical aspects related to this issue will
be discussed in the next section.

\section{Empirical analysis}
\label{sec:emp}
We are going to test the empirical performance of several instances of
the model introduced in Section \ref{sec:Hp}, that differ among each other
for the way they are calibrated and for some constraints on the
parameters $m$, $\sigma$, and $\alpha$.

The calibration of each instance of the model is done in the following
way: given a set of option prices observed at the same day and with
the same time to maturity, labeled from $1$ to $n$, for each
$j=1,\ldots,n$ we use the data with label
$(1,\ldots,j-1,j+1,\ldots,n)$ to calibrate the model, and with the
calibrated parameters we produce an estimate $\widehat{\pi}_j$ of the
price $\pi_j$ of the $j$-th option. The relative absolute pricing
error of $\widehat{\pi}_j$ with respect to $\pi_j$ is then defined
as $\abs{\widehat{\pi}_j/\pi_j-1}$.

Before describing each calibration method in detail and the
corresponding empirical performance, we briefly describe the data set
used.
We use S\&P500 index option data\footnote{The raw data are obtained
  from \emph{Historical Option Data}, see
  \texttt{www.historicaloptiondata.com}.} for the period January 3,
2012 to December 31, 2012.  During 2012 the annualized mean and
standard deviation of daily returns of the S\&P500 index were equal to
$11.09\%$ and $12.64\%$, respectively. During the same period the
1-year T-bill rate was very close to zero, with minimal variations: in
particular, its mean was equal to $0.16\%$, with a standard deviation
equal to $0.023\%$.  Our sample contains $77\,408$ observations of
European call and put options, $46\,854$ of which are put
options. Prices are averages of bid and ask prices. Data points with
time to maturity shorter than one day or volume less than $100$ are
eliminated. Descriptive statistics of the whole dataset are collected
in Table \ref{table01}.

\begin{table}[tbp]
\caption{Summary statistics for S\&P500 index options data}
\begin{center}
  \vspace{2mm} \parbox{\textwidth}{\footnotesize This table collects
    some simple statistics for prices of European call and put options
    on the S\&P500 index. The sample period is January 3, 2012 to
    December 31, 2012. Implied volatilities, expressed in percentage
    points, are annualized, time to maturity is expressed in days,
    strike and futures prices are expressed in index points.
    \vspace{4mm}} \footnotesize{{\begin{center} 
 \begin{tabular*}{\textwidth}{@{\extracolsep{\fill}}lccccccccc} 
 \toprule 
 \multicolumn{4}{c}{} & \multicolumn{5}{c}{Percentiles} & \multicolumn{1}{c}{}\\ 
 \cline{5-9}\\ 
 \multicolumn{4}{c}{}\\ 
Variable & Mean & Std& Min &$5\%$& $10\%$& $50\%$ & $90\%$ & $95\%$& Max  \\ 
 \midrule 
Call price             & 34.3& 98.8& 0.0& 0.1& 0.2& 9.2& 75.2& 115.5& 1270.0\\ 
Put price              & 21.3& 46.5& 0.0& 0.1& 0.1& 5.9& 58.2& 93.7& 1197.0\\ 
Implied $\sigma$       & 22.5& 11.5& 1.1& 11.9& 12.9& 19.3& 36.7& 44.9& 264.7\\ 
Implied ATM $\sigma$   & 22.0& 12.0& 1.1& 11.3& 12.1& 18.0& 38.6& 44.5& 202.8\\ 
Time to maturity & 96.7& 157.0& 1.0& 2.0& 4.0& 38.0& 269.0& 404.0& 1088.0\\ 
Strike price       & 1301.0& 208.4& 100.0& 950.0& 1075.0& 1345.0& 1480.0& 1525.0& 3000.0\\ 
Futures price       & 1374.4& 48.4& 1207.2& 1289.5& 1309.1& 1377.1& 1435.9& 1450.7& 1466.8\\ 
     \hline 
\bottomrule 
\end{tabular*} 
\end{center}

}}
\end{center}
\label{table01}
\end{table}

We focus on put options, and we eliminated from the dataset those
put options that (i) do not display price monotonicity with respect to
the strike price; (ii) have the same price and time to maturity but
different strike price.  In case (i) we eliminated options with low
trading volume breaking the monotonicity condition, and in case (ii)
we kept only the options with the highest and the lowest strike
prices.  This reduces the size of the sample to $43\,469$ put
contracts.  As is well known, index options on the S\&P500 are very
actively traded: the day with the largest number of unique put
contracts is December 21, 2012, that has $14$ expiration dates and
$269$ quoted put options prices (after the cleaning procedure
described above). The underlying price for this trading day was
$1\,430.20$ while the strike prices had values of $1\,100$, $1\,310$, and
$1\,425$ at the 10\textsuperscript{th}, 50\textsuperscript{th}, and
90\textsuperscript{th} percentile, respectively, with 93\% of the
contracts in the money. The time to maturity ranges from 4 days to
almost 3 years, in line with most other trading days.

\subsection{A simplified model}
The simplest calibration method that we consider slightly simplifies
the setting of Section \ref{sec:Hp}, assuming that there exists a constant
$\sigma_0>0$ such that
\[
\sigma = \sigma_0 \sqrt{t}, \qquad m := -\frac12 \sigma_0^2 t,
\]
where $t$ denotes time to maturity. Note that this can be considered
as a perturbation of the Black-Scholes model, where the standard
Gaussian density of suitably normalized returns is replaced by a
finite linear combination of (scaled) Hermite polynomials. In fact, in
the degenerate case where such linear combination reduces to a
multiple of the Hermite polynomial of order zero, one recovers
precisely the Black-Scholes model. Throughout this subsection we shall
write $\sigma$ in place of $\sigma_0$ for simplicity. The model's
calibration can thus be formulated as the minimization problem
\[
  \inf_{\substack{\sigma>0\\ \alpha \in \erre^{1+N}}}
  \norm{\Psi(\sigma)\alpha - 1}
  = \inf_{\sigma>0} \inf_{\alpha \in \erre^{1+N}}
  \norm{\Psi(\sigma)\alpha - 1},
\]
where $\Psi$ is the matrix defined in equation \eqref{eq:Psi} and
$\norm{\cdot}$ is a norm on $\erre^{n+1}$.
The first calibration technique that we consider starts, for any
$\sigma>0$, with the minimization problem
\begin{equation}
  \label{eq:ipsia}
  \inf_{\alpha \in \erre^{1+N}} \norm[\big]{\Psi(\sigma)\alpha -
    1}_{\ell^2},
\end{equation}
which can be solved by the standard ordinary least squares method to
provide a minimum point $\alpha_* = \alpha_*(\sigma)$, as discussed in
Section \ref{sec:cal}. Let
$E\colon \mathopen]0,\infty\mathclose[ \to \erre_+$ be the function
defined by
\[
E(\sigma) := \norm[\big]{\Psi(\sigma)\alpha_*(\sigma)-1}_{\ell^1},
\]
and consider the minimization problem
\[
  \inf_{\sigma>0} E(\sigma).
\]
Assuming that a minimum point $\sigma_*$ exists, we take $\sigma^*$
and $\alpha^*(\sigma^*)$ as estimates of the parameters of the model.
The calibration procedure thus obtained will be referred to as
procedure $\mathrm{H}_\sigma$.
The model produces pricing estimates that are consistently better than
the standard Black-Scholes one for every $N=1,\ldots,5$, in the sense
that the $10\%$, $25\%$, $50\%$, $75\%$, $90\%$, and $95\%$, quantiles
of the relative pricing error empirical distribution are smaller (up
to the $75\%$ quantiles they are around $50\%$ smaller). The pricing
error considerably improves with $N=2$, remains essentially unchanged
with $N=3$, and improves again quite drastically with $N=4$, to remain
again unchanged with $N=5$. The numerical results indicate that
Hermite approximations truncated at even degree $N$ are likely to be a
better choice, at least in the setting of calibration procedure
$\mathrm{H}_\sigma$. It should however be remarked that the pricing
performance of Black-Scholes with interpolated implied volatility is
still much better. Another important observation is that the size and
frequency of large pricing errors increase with $N$, consistently with
the ``conventional wisdom'' according to which the use of more and
more basis functions may cause numerical instability. Finally, the
relative error of Hermite pricing is particularly pronounced for
options with strike price lying far away from the strike prices of
observed options. This is checked by computing relative pricing errors
only for those options with strike $k$ such that
$k_\mathrm{min} < k < k_\mathrm{max}$, where $k_\mathrm{min}$ and
$k_\mathrm{max}$ are the smallest and the largest strike prices,
respectively, of the options used for calibration. One finds that
higher quantiles of the error distribution decrease considerably
(cf.~Table~\ref{tab:novAMP}). This observation is consistent with
approximations of densities by Hermite polynomials being usually good
around the center of the density, but not much so in the tails, where
they could even become negative (see, e.g., \cite{Kola} for a more
complete discussion). For this reason one cannot really expect good
approximate pricing for options that are deep out of the money, unless
prices of options with comparable strike prices are observed.
A further natural idea to try to limit the occasional large pricing
errors is to constrain the calibrated density to integrate to one and
to satisfy an approximate martingale condition, as in equation
\eqref{eq:vincoli}. In particular, the calibration procedure resulting
from adding these constraints to $\mathrm{H}_\sigma$, i.e. replacing
equation \eqref{eq:ipsia} by
\[
  \inf_{\alpha \in \Theta(\sigma)}
  \norm[\big]{\Psi(\sigma)\alpha - 1}_{\ell^2},
\]
where $\Theta$ accounts for the constraints mentioned above, as
discussed in Section \ref{sec:cal}, is labeled $\mathrm{H}_\sigma^{c,2}$.
Numerical results, however, are discouraging
(cf.~Table~\ref{tab:novAMP}), and suggest that the extra computational
burden is not worth.
\begin{rmk}
  The martingale condition is equivalent, under the present
  assumptions, to
  \begin{equation}
    \label{eq:lap}
    \E \exp \bigl( \sigma \sqrt{t} X - \frac12 \sigma^2 t \bigr) = 1.
  \end{equation}
  Recalling that $\E e^{\lambda X} = e^{\lambda^2/2}$ for every
  $\lambda \in \erre$ if and only if $X$ is a standard Gaussian random
  variable, it follows that if equation \eqref{eq:lap} is fulfilled
  for every $t \geq 0$, then $X$ is Gaussian. However, we do not
  require equation \eqref{eq:lap} to be verified for all $t$, but just
  for certain choices of $t$. Furthermore, we should recall that $X$
  itself depends on $t$, so it does not necessarily have to be
  Gaussian.
\end{rmk}

\begin{table}[tbp]
\caption{Pricing errors Black and Scholes}
\begin{center}
  \vspace{2mm} \parbox{\textwidth}{\footnotesize The table reports
    selected quantiles of the distribution of empirical pricing errors
    (in percentage points) of the Black and Scholes estimators. Each
    column matches the corresponding one in the tables relative to
    Hermite pricing. Figures in parenthesis refer to pricing errors
    obtained excluding strikes outside the interval of observed
    ones. \vspace{4mm}} \footnotesize{{\begin{center} 
 \begin{tabular*}{\textwidth}{@{\extracolsep{\fill}}lccccccc} 
 \toprule 
 \multicolumn{6}{c}{}\\ 
 \multicolumn{6}{c}{\textsc{Empirical distribution of Pricing errors}}\\ 
 \multicolumn{6}{c}{}\\ 
\multicolumn{6}{c}{Black \& Scholes}\\ 
 \midrule 
Quantiles &$N=1$ &$N=2$ & $N=3$ & $N=4$ & $N=5$ \\  
 \cmidrule{2-6} 
$10\%$     & 1.3 (1.4)& 0.1 (0.2)& 0.2 (0.2)& 0.1 (0.1)& 0.1 (0.1)\\ 
$25\%$     & 4.2 (4.2)& 0.5 (0.5)& 0.6 (0.5)& 0.4 (0.4)& 0.4 (0.4)\\ 
$50\%$     & 10.5 (10.1)& 1.9 (1.8)& 2.0 (1.9)& 1.6 (1.5)& 1.6 (1.5)\\ 
$75\%$     & 22.7 (21.0)& 6.8 (6.0)& 6.8 (6.1)& 6.4 (5.5)& 6.5 (5.6)\\ 
$90\%$     & 45.9 (40.5)& 19.7 (15.9)& 19.5 (15.8)& 22.4 (16.4)& 23.7 (17.3)\\ 
$95\%$     & 63.1 (56.5)& 37.4 (27.0)& 38.0 (27.0)& 51.3 (30.8)& 59.8 (34.1)\\ 
 \multicolumn{6}{c}{}\\ 
\multicolumn{6}{c}{Black \& Scholes with linearly interpolated $\sigma$}\\ 
\midrule 
Quantiles &$N=1$ &$N=2$ & $N=3$ & $N=4$ & $N=5$ \\  
 \cmidrule{2-6}  
$10\%$     & 0.1 (0.1)& 0.1 (0.1)& 0.1 (0.1)& 0.1 (0.1)& 0.1 (0.1)\\ 
$25\%$     & 0.2 (0.2)& 0.2 (0.2)& 0.2 (0.2)& 0.2 (0.3)& 0.2 (0.3)\\ 
$50\%$     & 1.0 (1.0)& 1.0 (1.0)& 1.0 (1.0)& 1.0 (1.0)& 1.0 (1.1)\\ 
$75\%$     & 4.7 (4.4)& 4.7 (4.4)& 4.8 (4.5)& 4.8 (4.5)& 4.9 (4.6)\\ 
$90\%$     & 16.4 (13.9)& 16.2 (13.9)& 16.2 (14.0)& 16.3 (14.1)& 16.3 (14.2)\\ 
$95\%$     & 31.4 (25.1)& 31.0 (25.1)& 30.7 (25.1)& 30.3 (25.3)& 30.2 (25.3)\\ 
 \midrule 
Test points& 43469 (37760)& 42755 (37522)& 41815 (37052)& 40830 (36461)& 39834 (35797)\\ 
\bottomrule 
\end{tabular*} 
\end{center}
}}
\end{center}
\label{table03}
\end{table}

\begin{table}[tbp]
  \caption{Pricing errors for Hermite models $\mathrm{H}_\sigma$ and
    $\mathrm{H}_\sigma^{c,2}$}
\begin{center}
  \vspace{2mm} \parbox{\textwidth}{\footnotesize The table reports
    selected quantiles of the distribution of empirical pricing errors
    (in percentage points) of the Hermite estimators
    $\mathrm{H}_\sigma$ and $\mathrm{H}_\sigma^{c,2}$. Figures in
    parenthesis refer to pricing errors obtained excluding strikes
    outside the interval of observed ones. \vspace{4mm}}
  \footnotesize{{\begin{center} 
 \begin{tabular*}{\textwidth}{@{\extracolsep{\fill}}lccccccc} 
 \toprule 
 \multicolumn{6}{c}{}\\ 
 \multicolumn{6}{c}{\textsc{Empirical distribution of Pricing errors}}\\ 
 \multicolumn{6}{c}{}\\ 
 \multicolumn{6}{c}{$\mathrm{H}_{\sigma}$}\\ 
 \midrule 
Quantiles & $N=1$ &$N=2$ & $N=3$ & $N=4$ & $N=5$ \\  
 \cmidrule{2-6} 
$10\%$     & 2.9 (3.2)& 0.4 (0.4)& 0.4 (0.5)& 0.1 (0.1)& 0.1 (0.2)\\ 
$25\%$     & 8.8 (8.6)& 1.5 (1.5)& 1.6 (1.6)& 0.5 (0.5)& 0.5 (0.5)\\ 
$50\%$     & 19.4 (18.3)& 4.6 (4.4)& 5.3 (5.0)& 1.8 (1.7)& 1.8 (1.7)\\ 
$75\%$     & 37.2 (33.3)& 12.2 (11.0)& 13.3 (12.0)& 6.6 (5.8)& 6.5 (5.7)\\ 
$90\%$     & 66.7 (60.8)& 27.9 (23.5)& 28.3 (23.5)& 21.0 (16.3)& 21.7 (16.5)\\ 
$95\%$     & 80.0 (74.8)& 47.5 (36.2)& 46.5 (35.3)& 44.9 (28.7)& 50.3 (30.5)\\ 
 \multicolumn{6}{c}{}\\ 
 \multicolumn{6}{c}{$\mathrm{H}_{\sigma}^{c,2}$}\\ 
 \midrule 
Quantiles & $N=1$ &$N=2$ & $N=3$ & $N=4$ & $N=5$ \\  
 \cmidrule{2-6} 
$10\%$     & 9.6 (10.8)& 10.1 (11.3)& 0.8 (0.9)& 0.6 (0.7)& 0.5 (0.6)\\ 
$25\%$     & 20.6 (21.5)& 23.5 (24.4)& 2.5 (2.6)& 2.2 (2.3)& 1.8 (1.8)\\ 
$50\%$     & 36.0 (36.8)& 38.7 (39.3)& 6.2 (6.3)& 6.7 (6.7)& 5.0 (4.7)\\ 
$75\%$     & 79.0 (79.3)& 75.0 (75.0)& 15.1 (14.6)& 15.7 (15.2)& 11.8 (10.5)\\ 
$90\%$     & 93.3 (92.9)& 92.3 (91.7)& 45.1 (37.0)& 34.5 (31.4)& 24.9 (20.4)\\ 
$95\%$     & 96.3 (96.0)& 95.7 (95.3)& 66.7 (64.5)& 56.4 (50.0)& 39.5 (30.9)\\ 
\midrule 
Test points& 43469 (37760)& 42755 (37522)& 41815 (37052)& 40830 (36461)& 39834 (35797)\\ 
\bottomrule 
\end{tabular*} 
\end{center}
}}
\end{center}
\label{tab:novAMP}
\end{table}

The calibration of model $\mathrm{H}_\sigma$ discussed so far is
somewhat inconsistent because it ``mixes'' the $\ell^2$ and the
$\ell^1$ norms. It is then natural to ask whether a consistent use of
the $\ell^1$ norm, i.e. of least absolute deviations, would improve
the statistics of relative pricing error. It turns out that this is
hardly the case, with empirical results suggesting that the (relative)
accuracy of procedure $\mathrm{H}_\sigma$ is very
satisfying. Moreover, the method of ordinary least squares is very
fast and less prone to numerical instability in comparison to the
method of least absolute deviations.  In order to substantiate these
claims, let us introduce further calibration procedures: if equation
\eqref{eq:ipsia} is replaced by
\begin{equation*}
  \inf_{\alpha \in \erre^{1+N}} \norm[\big]{\Psi(\sigma)\alpha -
    1}_{\ell^1},
\end{equation*}
the resulting procedure is labeled $\mathrm{H}^1_\sigma$.
Consider now the (numerical) minimization problem
\[
  \inf_{\substack{\sigma>0\\ \alpha \in \erre^{1+N}}}
  \norm[\big]{\Psi(\sigma)\alpha - 1}_{\ell^1}
\]
with starting point $(\sigma_{\mathrm{BS}},\alpha_{\mathrm{BS}})$,
where $\sigma_{\mathrm{BS}}$ is such that the $\ell^1$ distance
between observed option prices and Black-Scholes prices with
volatility $\sigma_{\mathrm{BS}}$ is minimized, and
$\alpha_{\mathrm{BS}} = (1/\sqrt{2\pi},0,\ldots,0)$. The resulting
calibration procedure is labeled $\mathrm{H}^{1,0}_\sigma$. If the
initial point for the minimization algorithm is chosen as the minimum
point of the $\mathrm{H}_\sigma$ procedure, the resulting procedure is
labeled $\mathrm{H}^{1,2}_\sigma$.
Note that, due to the lack of convexity of the function
$(\sigma,\alpha) \mapsto \norm{\Phi(\sigma)\alpha - 1}$, numerical
minimization algorithms are only expected to converge to a local
minimum around the initial point $(\sigma_0,\alpha_0)$, for which
there appears to be no ``canonical'' choice. Procedure
$\mathrm{H}_\sigma^{1,0}$ amounts to looking for a Hermite model
minimizing the $\ell^1$ error starting its search on the
``degenerate'' Hermite model of order zero, i.e. from the
Black-Scholes model. Similarly, procedure $\mathrm{H}_\sigma^{1,2}$
looks for a local minimum point around the optimal solution provided
by $\mathrm{H}_\sigma$. It is perhaps useful to recall that in both
procedures $\mathrm{H}_\sigma$ and $\mathrm{H}_\sigma^1$ the
minimization step in $\sigma$ can be done with numerical algorithms
that require just an upper and a lower bound, rather than a starting
point.

Numerical results on our dataset indicate that
\begin{itemize}
  \setlength{\itemsep}{0pt}
\item[(a)] $\mathrm{H}_\sigma^1$ performs slightly better than
  $\mathrm{H}_\sigma$ at the level of lower quantiles of the error
  distribution (up to $50\%$), and slightly worse at the level of
  higher quantiles, with the slight advantage reducing as the order $N$
  of the Hermite approximation increases;
\item[(b)] The performance of $\mathrm{H}_\sigma^{1,0}$ is overall
  comparable to the ones of both $\mathrm{H}_\sigma$ and
  $\mathrm{H}_\sigma^1$ for values of $N$ up to three, while it is
  clearly worse for values of $N=4$ and $N=5$;
\item[(c)] the minimum point of $\mathrm{H}_\sigma^{1,2}$ is
  consistently very close to the one of $\mathrm{H}_\sigma$, and,
  accordingly, the improvement in pricing error is very small across
  all values of $N$ and percentiles of the error
  distribution. Moreover, the distribution of pricing error becomes
  almost indistinguishable from the one of $\mathrm{H}^1_\sigma$ as
  $N$ increases (cf.~Table~\ref{tab:lad}).
\end{itemize}

\begin{table}[tbp]
  \caption{Pricing errors of Hermite models $\mathrm{H}_{\sigma}^{1}$,
    $\mathrm{H}_{\sigma}^{1,0}$, and $\mathrm{H}_{\sigma}^{1,2}$}
\begin{center}
  \vspace{2mm} \parbox{\textwidth}{\footnotesize The table reports
    selected quantiles of the distribution of the empirical pricing
    errors (in percentage points) of the Hermite estimators
    $\mathrm{H}_{\sigma}^{1}$, $\mathrm{H}_{\sigma}^{1,0}$, and
    $\mathrm{H}_{m,\sigma}^{1,2}$. Figures in parenthesis refer to
    pricing errors obtained excluding strikes outside the interval of
    observed ones.  \vspace{4mm}}
  \footnotesize{{\begin{center} 
 \begin{tabular*}{\textwidth}{@{\extracolsep{\fill}}lccccccc} 
 \toprule 
 \multicolumn{6}{c}{}\\ 
 \multicolumn{6}{c}{\textsc{Empirical distribution of Pricing errors}}\\ 
 \multicolumn{6}{c}{}\\ 
 \multicolumn{6}{c}{$\mathrm{H}_{\sigma}^{1}$}\\ 
 \midrule 
Quantiles & $N=1$ &$N=2$ & $N=3$ & $N=4$ & $N=5$ \\  
 \cmidrule{2-6} 
$10\%$     & 3.1 (3.3)& 0.4 (0.4)& 0.4 (0.4)& 0.1 (0.1)& 0.1 (0.1)\\ 
$25\%$     & 7.6 (7.5)& 1.3 (1.3)& 1.5 (1.5)& 0.4 (0.4)& 0.5 (0.5)\\ 
$50\%$     & 17.2 (16.1)& 4.3 (4.1)& 5.0 (4.7)& 1.7 (1.6)& 1.8 (1.7)\\ 
$75\%$     & 38.6 (34.3)& 12.2 (10.8)& 13.4 (12.1)& 6.5 (5.8)& 6.4 (5.7)\\ 
$90\%$     & 69.1 (66.6)& 31.0 (25.9)& 29.9 (24.9)& 21.5 (16.5)& 21.8 (16.6)\\ 
$95\%$     & 83.3 (79.1)& 52.2 (42.8)& 50.0 (38.7)& 47.2 (30.4)& 50.4 (32.0)\\ 
 \multicolumn{6}{c}{}\\ 
 \multicolumn{6}{c}{$\mathrm{H}_{\sigma}^{1,0}$}\\ 
 \midrule 
Quantiles & $N=1$ &$N=2$ & $N=3$ & $N=4$ & $N=5$ \\  
 \cmidrule{2-6} 
$10\%$     & 2.8 (3.0)& 0.5 (0.5)& 0.5 (0.5)& 0.6 (0.6)& 0.6 (0.6)\\ 
$25\%$     & 7.0 (6.9)& 1.6 (1.5)& 1.5 (1.5)& 1.8 (1.8)& 1.8 (1.8)\\ 
$50\%$     & 15.8 (14.8)& 4.9 (4.6)& 4.6 (4.3)& 5.3 (5.0)& 5.3 (4.9)\\ 
$75\%$     & 38.1 (33.5)& 13.4 (11.9)& 13.6 (11.9)& 15.5 (13.7)& 15.5 (13.7)\\ 
$90\%$     & 72.4 (66.7)& 31.5 (26.1)& 40.1 (31.9)& 50.0 (42.7)& 50.0 (40.9)\\ 
$95\%$     & 84.3 (80.4)& 50.2 (41.9)& 66.7 (54.7)& 73.9 (66.7)& 74.6 (66.7)\\ 
 \multicolumn{6}{c}{}\\ 
 \multicolumn{6}{c}{$\mathrm{H}_{\sigma}^{1,2}$}\\ 
 \midrule 
Quantiles & $N=1$ &$N=2$ & $N=3$ & $N=4$ & $N=5$ \\  
 \cmidrule{2-6} 
$10\%$     & 3.1 (3.3)& 0.3 (0.4)& 0.4 (0.4)& 0.1 (0.1)& 0.1 (0.2)\\ 
$25\%$     & 7.8 (7.6)& 1.3 (1.3)& 1.5 (1.5)& 0.5 (0.5)& 0.5 (0.5)\\ 
$50\%$     & 17.5 (16.3)& 4.3 (4.1)& 5.1 (4.8)& 1.8 (1.7)& 1.8 (1.7)\\ 
$75\%$     & 38.5 (34.2)& 12.1 (10.8)& 13.4 (12.0)& 6.7 (5.9)& 6.6 (5.7)\\ 
$90\%$     & 67.0 (64.2)& 29.6 (24.7)& 29.3 (24.3)& 21.4 (16.6)& 22.0 (16.6)\\ 
$95\%$     & 81.4 (76.7)& 49.7 (38.6)& 47.6 (36.6)& 45.0 (29.2)& 50.3 (30.9)\\ 
\midrule 
Test points& 43469 (37760)& 42755 (37522)& 41815 (37052)& 40830 (36461)& 39834 (35797)\\ 
\bottomrule 
\end{tabular*} 
\end{center}
}}
\end{center}
\label{tab:lad}
\end{table}

These empirical observations suggest that, in spite of its theoretical
inconsistency, procedure $\mathrm{H}_\sigma$ is not necessarily worse
than the sounder procedure $\mathrm{H}_\sigma^1$. One should also take
into account that, even though least absolute deviation is more robust
to outliers than ordinary least squares, standard numerical routines
for the former did not run nearly as smoothly as those for the latter
in our dataset (see Appendix~B for more detail about the numerical
implementation of $\mathrm{H}_\sigma^1$ via linear programming, as
outlined in Section \ref{sec:cal}).
Moreover, the rather simple-minded procedure $\mathrm{H}_\sigma^{1,0}$
turns out to be a viable alternative for lower values of $N$, even
though it is clearly considerably slower than $\mathrm{H}_\sigma$ and
$\mathrm{H}_\sigma^1$, as it involves the minimization of a function
on a higher-dimensional space. It seems interesting to observe that
the lack of convexity mentioned above appears to have a considerable
negative impact on the pricing error only for values of $N$ larger
than three. It is natural to speculate that, as the dimension of the
state space over which the objective function is minimized increases,
more and more local minima appear.

\subsection{Analysis of the full Hermite model}
We now turn to examining the empirical performance of the full model
introduced in Section \ref{sec:Hp}.  The simplest calibration procedure,
labeled $\mathrm{H}_{m,\sigma}$, consists in the minimization problem
\begin{equation}
  \label{eq:tsetse}
  \inf_{\substack{m \in \erre \\ \sigma \in \mathopen]0,\infty\mathclose[}}
  \norm[\big]{\Psi(m,\sigma)\alpha_*(m,\sigma)-1}_{\ell^1},
\end{equation}
where, for any real numbers $m$ and $\sigma$, with $\sigma>0$,
$\alpha_*(m,\sigma)$ is a minimum point of the convex minimization
problem
\begin{equation}
  \label{eq:mips}
  \inf_{\alpha \in \erre^{1+N}} \norm[\big]{\Psi(m,\sigma)\alpha-1}_{\ell^2}.
\end{equation}
The starting point for the numerical minimization algorithm over $m$
and $\sigma$ is chosen as the minimum point of calibration procedure
$H_\sigma$. More precisely, if $(\sigma_*,\alpha_*)$ is the
calibration produced by $H_\sigma$, the initial point for the
numerical solution of equation \eqref{eq:tsetse} is
\[
  m_0 := -\frac12 \sigma_*^2 t, \qquad
  \sigma_0 := \sigma_* \sqrt{t},
\]
where $t$ is the time to maturity. Adding the constraints in equation
\eqref{eq:vincoli} to \eqref{eq:mips} produces the calibration
procedure labeled $\mathrm{H}_{m,\sigma}^{c,2}$. In this case the
numerical solution of equation \eqref{eq:tsetse} takes as starting
point the minimum point obtained by calibration procedure
$\mathrm{H}_{\sigma}^{c,2}$, in the same sense already discussed
above.

Empirical results (see Table~\ref{tab:novAMPms}) show that the
extra degree of freedom of $\mathrm{H}_{m,\sigma}$ with respect to
$\mathrm{H}_\sigma$ produces massive improvements in pricing accuracy
only for $N \leq 3$, and a more modest improvement with $N=4$ and
$N=5$. In particular, for $N \leq 3$, all quantiles of the error
distribution up to $95\%$ are lower than the corresponding quantiles
for the models in the previous subsection. For $N=4$ and $N=5$,
quantiles up to $75\%$ improve, but become worse at higher
levels. This is not too surprising considering that extra parameters
tend to improve accuracy but to worsen stability.
On the other hand, the improvement of $\mathrm{H}_{m,\sigma}^{c,2}$
with respect $\mathrm{H}_\sigma^{c,2}$ is very strong for all values
of $N$, to the point that, for $N=5$, its performance is not much
worse than the ones of $\mathrm{H}_\sigma$ and
$\mathrm{H}_{m,\sigma}$. Moreover, the large errors produced by
$\mathrm{H}_\sigma^{c,2}$ for $N=5$ are considerably smaller than
those of other procedures. However, empirical observations already
made in the previous subsection are confirmed: passing from
$\mathrm{H}_{m,\sigma}$ to $\mathrm{H}_{m,\sigma}^{c,2}$ reduces the
number of large errors in some cases, but does not improve the
precision: the error distribution of $\mathrm{H}_{m,\sigma}^{c,2}$
dominates the one of $\mathrm{H}_{m,\sigma}$ up to the $75\%$ quantile
across all values of $N$.

An important numerical observation is that the estimated values of
$\alpha$ are often enormous (of order of magnitude $10^{150}$). Even
though such values can hardly be interpreted, they do not compromise,
in the overwhelming majority of cases, neither calibration error nor
pricing error. Perhaps somewhat surprisingly, at least from the point
of view of numerical stability, adding lower and upper bounds to
equation \eqref{eq:mips} produces \emph{worse} results (numerical output
relative to these attempts is not reproduced).
On the other hand, in the case of calibration procedure
$\mathrm{H}_{m,\sigma}^{c,2}$ the minimization problem \eqref{eq:mips}
subject to the additional constraints \eqref{eq:vincoli} is solved
numerically using quadratic programming, for which, to avoid numerical
crashes, it was necessary to constrain $\abs{\alpha}$ to be less than
the inverse of machine precision. This bound however is never reached,
and estimates of $\alpha$ are in this case much better behaved. On the
other hand, as already remarked, the calibration without constraints
displays better pricing accuracy in the large majority of cases.

The calibration procedure obtained replacing the $\ell^2$ norm in
equation \eqref{eq:mips} by the $\ell^1$ norm, which would naturally
be labeled $\mathrm{H}_{m,\sigma}^1$, turns out to be numerically very
unstable on our dataset, with minimization by linear programming, via
the GLPK routines, crashing too often to be usable. Roughly speaking,
the reason is that the matrix $\Psi(m,\sigma)$ becomes very singular
and the numerical linear programming routines break down. For this
reason, whenever $\Psi(m,\sigma)$ is too ``large'' (see Appendix B for
detail), we use instead the estimates produced by the procedures
$\mathrm{H}_{m,\sigma}^{1,0}$ and $\mathrm{H}_{m,\sigma}^{1,2}$, that
correspond to the minimization of the function
$(m,\sigma,\alpha) \mapsto \norm{\Psi(m,\sigma)\alpha-1}_{\ell^1}$
over the set
$\erre \times \mathopen]0,\infty\mathclose[ \times \erre^{1+N}$, using
as starting point the Black-Scholes parameters and the $H_\sigma$
parameters, respectively. With a slight abuse of notation, the
procedures so obtained are still labeled $\mathrm{H}_{m,\sigma}^{1,0}$
and $\mathrm{H}_{m,\sigma}^{1,2}$, respectively.  Note that it would
not make sense to use as starting point the parameters calibrated by
$\mathrm{H}_{m,\sigma}$ for the reasons discussed above.

The empirical results reported in Table~\ref{tab:ladms} show that
least absolute deviation estimates starting from the Black-Scholes
parameters are no longer comparable to the estimates produced by the
two-step OLS optimization (i.e. by $\mathrm{H}_{m,\sigma}$), even for
lower values of $N$. On the other hand, the performance of the
$\mathrm{H}_{m,\sigma}^{1,2}$ procedure is indeed comparable to the
one of $\mathrm{H}_{m,\sigma}$ for $N=4$ and $N=5$, but it does not
offer any worthy advantage, apart from the size of $\alpha$. In fact,
one should take into account that, for reasons already discussed
above, the minimization algorithm used by
$\mathrm{H}_{m,\sigma}^{1,2}$ is much slower than the two-step
procedure of $\mathrm{H}_{m,\sigma}$. Moreover,
$\mathrm{H}_{m,\sigma}^{1,2}$ has a performance that is only slightly
better than the one of $\mathrm{H}_{\sigma}^{1,2}$ in the range $N=2$
to $N=4$, and essentially identical for $N=5$ for percentiles up to
50\%. Therefore, also by considerations of computational complexity,
it does not seem particularly interesting. The results, however, are
important in the sense that they confirm the good empirical
performance of our proposed procedure $\mathrm{H}_{m,\sigma}$, in
spite of its theoretical inconsistency.

\begin{table}[tbp]
  \caption{Pricing errors of Hermite models $\mathrm{H}_{m,\sigma}$
    and $\mathrm{H}_{m,\sigma}^{c,2}$}
\begin{center}
  \vspace{2mm} \parbox{\textwidth}{\footnotesize The table reports
    selected quantiles of the distribution of the empirical pricing
    errors (in percentage points) of the Hermite estimators
    $\mathrm{H}_{m,\sigma}$ and $\mathrm{H}_{m,\sigma}^{c,2}$. Figures
    in parenthesis refer to pricing errors obtained excluding strikes
    outside the interval of observed ones.  \vspace{4mm}}
  \footnotesize{{\begin{center} 
 \begin{tabular*}{\textwidth}{@{\extracolsep{\fill}}lccccccc} 
 \toprule 
 \multicolumn{6}{c}{}\\ 
 \multicolumn{6}{c}{\textsc{Empirical distribution of Pricing errors}}\\ 
 \multicolumn{6}{c}{}\\ 
 \multicolumn{6}{c}{$\mathrm{H}_{m,\sigma}$}\\ 
 \midrule 
Quantiles & $N=1$ &$N=2$ & $N=3$ & $N=4$ & $N=5$ \\  
 \cmidrule{2-6} 
$10\%$     & 1.3 (1.4)& 0.1 (0.2)& 0.2 (0.2)& 0.1 (0.1)& 0.1 (0.1)\\ 
$25\%$     & 4.2 (4.2)& 0.5 (0.5)& 0.6 (0.5)& 0.4 (0.4)& 0.4 (0.4)\\ 
$50\%$     & 10.5 (10.1)& 1.9 (1.8)& 2.0 (1.9)& 1.6 (1.5)& 1.6 (1.5)\\ 
$75\%$     & 22.7 (21.0)& 6.8 (6.0)& 6.8 (6.1)& 6.4 (5.5)& 6.5 (5.6)\\ 
$90\%$     & 45.9 (40.5)& 19.7 (15.9)& 19.5 (15.8)& 22.4 (16.4)& 23.7 (17.3)\\ 
$95\%$     & 63.1 (56.5)& 37.4 (27.0)& 38.0 (27.0)& 51.3 (30.8)& 59.8 (34.1)\\ 
 \multicolumn{6}{c}{}\\ 
 \multicolumn{6}{c}{$\mathrm{H}_{m,\sigma}^{c,2}$}\\ 
 \midrule 
Quantiles & $N=1$ &$N=2$ & $N=3$ & $N=4$ & $N=5$ \\  
 \cmidrule{2-6} 
$10\%$     & 2.5 (2.6)& 1.7 (1.8)& 0.4 (0.5)& 0.3 (0.4)& 0.2 (0.2)\\ 
$25\%$     & 6.2 (6.0)& 4.5 (4.4)& 1.5 (1.5)& 1.2 (1.2)& 0.6 (0.6)\\ 
$50\%$     & 13.8 (12.8)& 10.7 (10.0)& 4.5 (4.2)& 3.6 (3.4)& 2.2 (2.1)\\ 
$75\%$     & 31.1 (27.5)& 24.9 (22.4)& 12.0 (10.6)& 10.2 (9.3)& 7.1 (6.5)\\ 
$90\%$     & 64.4 (57.7)& 53.2 (48.7)& 28.4 (23.2)& 24.4 (20.6)& 19.5 (17.0)\\ 
$95\%$     & 77.6 (72.9)& 67.6 (66.0)& 50.0 (37.9)& 41.5 (32.4)& 35.5 (29.7)\\ 
 \midrule 
Test points& 43469 (37760)& 42755 (37522)& 41815 (37052)& 40830 (36461)& 39958 (35921)\\ 
\bottomrule 
\end{tabular*} 
\end{center}
}}
\end{center}
\label{tab:novAMPms}
\end{table}

\begin{table}[tbp]
  \caption{Pricing errors of Hermite models $\mathrm{H}_{m,\sigma}^{1,0}$
    and $\mathrm{H}_{m,\sigma}^{1,2}$}
\begin{center}
  \vspace{2mm} \parbox{\textwidth}{\footnotesize The table reports
    selected quantiles of the distribution of the empirical pricing
    errors (in percentage points) of the Hermite estimators
    $\mathrm{H}_{m,\sigma}^{1,0}$ and $\mathrm{H}_{m,\sigma}^{1,2}$. Figures
    in parenthesis refer to pricing errors obtained excluding strikes
    outside the interval of observed ones.  \vspace{4mm}}
  \footnotesize{{\begin{center} 
 \begin{tabular*}{\textwidth}{@{\extracolsep{\fill}}lccccccc} 
 \toprule 
 \multicolumn{6}{c}{}\\ 
 \multicolumn{6}{c}{\textsc{Empirical distribution of Pricing errors}}\\ 
 \multicolumn{6}{c}{}\\ 
 \multicolumn{6}{c}{$\mathrm{H}_{m,\sigma}^{1,0}$}\\ 
 \midrule 
Quantiles & $N=1$ &$N=2$ & $N=3$ & $N=4$ & $N=5$ \\  
 \cmidrule{2-6} 
$10\%$     & 2.2 (2.3)& 0.8 (0.9)& 0.5 (0.6)& 0.9 (1.1)& 0.5 (0.6)\\ 
$25\%$     & 5.4 (5.3)& 2.5 (2.4)& 1.8 (1.7)& 3.2 (3.2)& 1.7 (1.6)\\ 
$50\%$     & 12.9 (11.9)& 7.0 (6.4)& 5.8 (5.3)& 8.7 (8.2)& 5.0 (4.6)\\ 
$75\%$     & 30.3 (26.7)& 18.7 (16.2)& 17.7 (15.8)& 22.0 (19.5)& 14.1 (12.4)\\ 
$90\%$     & 63.0 (56.5)& 46.9 (39.0)& 47.8 (40.0)& 51.5 (46.4)& 39.2 (32.2)\\ 
$95\%$     & 76.2 (71.0)& 66.7 (61.7)& 66.7 (61.6)& 67.5 (65.8)& 66.7 (55.6)\\ 
 \multicolumn{6}{c}{}\\ 
 \multicolumn{6}{c}{$\mathrm{H}_{m,\sigma}^{1,2}$}\\ 
 \midrule 
Quantiles & $N=1$ &$N=2$ & $N=3$ & $N=4$ & $N=5$ \\  
 \cmidrule{2-6} 
$10\%$     & 2.5 (2.6)& 0.4 (0.4)& 0.4 (0.4)& 0.1 (0.1)& 0.1 (0.1)\\ 
$25\%$     & 6.2 (6.0)& 1.3 (1.3)& 1.4 (1.4)& 0.4 (0.4)& 0.5 (0.5)\\ 
$50\%$     & 14.5 (13.4)& 4.2 (4.0)& 4.7 (4.4)& 1.7 (1.6)& 1.8 (1.7)\\ 
$75\%$     & 34.3 (30.1)& 12.0 (10.7)& 12.9 (11.5)& 6.5 (5.7)& 6.9 (6.0)\\ 
$90\%$     & 66.7 (60.0)& 30.5 (25.4)& 29.4 (24.6)& 21.5 (16.5)& 25.6 (19.0)\\ 
$95\%$     & 79.6 (73.8)& 50.5 (41.3)& 48.9 (37.7)& 46.7 (30.0)& 66.7 (40.4)\\ 
\midrule 
Test points& 43469 (37760)& 42755 (37522)& 41815 (37052)& 40830 (36461)& 39834 (35797)\\ 
\bottomrule 
\end{tabular*} 
\end{center}
}}
\end{center}
\label{tab:ladms}
\end{table}

\section{Empirical analysis on synthetic data}
\label{sec:emps}
We are going to describe the results of an empirical analysis,
analogous to the one described in the previous section, on a set of
synthetic data, generated using Hermite processes (see
Appendix~\ref{app:HermProc} for basic definitions and results, and,
e.g., \cite{Ra:Herm} and references therein for financial
applications).\footnote{We thank the referee for the suggestion to
  consider data generated by Hermite processes.}
Such an analysis can be considered as a sort of empirical robustness
test, as a financial interpretation along the lines described in
previous sections is, in general, not possible. More precisely, we
shall produce synthetic data of the type
\begin{equation}
  \label{eq:sinte}
  \pi(k) = \int_\erre \bigl( k - Y_0e^{\sigma x + m} \bigr)^+\,dF(x),
\end{equation}
where \(Y_0>0\), \(\sigma>0\) and \(m\) are constants, and \(F\) is the
distribution function of a (non-Gaussian) Hermite process at time one.
The values \(\pi(k)\), however, cannot be interpreted as prices of
options in a Hermite market, as Hermite processes are not
semimartingales, hence the standard pricing methods in terms of
expectations under a risk-neutral measure do not make sense any
longer.

On the other hand, the problem of estimating \(\pi(k)\) (or, more
generally, of estimating \(F\), as explained in \S\ref{ssec:equiv})
from a finite set of observations \((\pi(k_i))_{i \in I}\) is
meaningful for \emph{any} distribution function \(F\), independently
of any financial interpretation. It is in this sense that the
numerical results obtained should be interpreted as a sort of
robustness test.

\medskip

We produced synthetic values of \(\pi(k)\), as defined by
\eqref{eq:sinte}, with the parameters \(k\), \(Y_0\), \(\sigma\) and
\(m\) chosen in terms of the dataset considered in the previous
section.
As a first step, we randomly selected 25 days from the dataset. For
each day there are ``blocks'' of options with the same time to
maturity. Let us now consider a day and a block fixed: we set \(Y_0\)
equal to the price of the underlying \(S_0\), and denoting the time to
maturity and the calibrated Black-Scholes implied volatility for the
block under consideration by \(t\) and \(\sigma_0\), respectively, we
set
\begin{equation}
  \label{eq:parsint}
  \sigma = \sigma_0 \sqrt{t}, \qquad
  m = -\frac12 \sigma^2_0 t.
\end{equation}
Furthermore, random samples of a Hermite process with parameters
\(k=3\) and \(H=0.63\) evaluated at time one, denoted by
\(Z^3_{0.63}(1)\), are generated using the weak convergence results
gathered in Appendix~\ref{app:HermProc}.\footnote{More precisely, one
  should say that the simulated random samples are only in the domain
  of attraction of the distribution of the random variable
  \(Z^3_{0.63}(1)\) -- see Appendix~\ref{app:HermProc} for more
  detail.} The empirical distribution function of the set of simulated
random samples is denoted by \(F\). Finally, we computed \(\pi(k)\) as
in \eqref{eq:sinte} for the values of \(k\) corresponding to the
strike prices in the block under consideration in the original
dataset. The whole procedure is repeated for each block of each day,
thus obtaining a synthetic dataset that has approximately 10\% the
size of the real dataset used in the previous section.  The empirical
analysis described in the previous section is then applied to the
synthetic data thus produced.

Before describing the results of the analysis, some remarks are in
order. The choice of the parameters \(\sigma\) and \(m\)
(see~\eqref{eq:parsint} above) is guided simply by an analogy to the
case discussed in the previous section. In this regard it is probably
worth mentioning that the process
\(Y_t = Y_0 \exp\bigl( \sigma Z^k_H(t) \bigr)\) does \emph{not} have,
in general, finite expectation, as elements of the \(n\)-th Wiener
chaos, with \(n \geq 3\), do not admit any exponential moments
(see~\cite[Corollary~6.13]{Janson}). However, since
\(0 \leq (k-e^x)^+ \leq k\) for every \(x \in \erre\), the
expectations \(\E(k-Y_t)^+\) are always finite. Analogously, the
distribution of \(Z^k_H(1)\) is not expected to have a density in
\(L^2(\erre)\) (see Appendix~\ref{app:dens} for more detail). However,
the empirical distribution function \(F\) is compactly supported and
bounded, hence (a smoothed version of) its density is certainly in
\(L^2(\erre)\).

\medskip

Let us now discuss the empirical results obtained on the synthetic
dataset, on which we have applied the estimation methods
\(\mathrm{H}_\sigma\), \(\mathrm{H}_{m,\sigma}\),
\(\mathrm{H}^{1,0}_\sigma\), and \(\mathrm{H}^{1,2}_\sigma\), in
addition to the Black-Scholes methods with implied volatility and with
interpolation on the implied volatility curve (to which we shall refer
as BS and \(\mathrm{BS_i}\), respectively). Methods involving least
absolute deviation techniques implemented via linear programming have
been excluded because of their numerical instability (see the
corresponding remarks in the previous section and
Appendix~\ref{app:num}). Similarly, the constrained methods would not
make sense in the present setting, as the synthetic data cannot be
interpreted as prices, as already discussed, hence the approximate
martingale property would just be a spurious constraint.\footnote{In
  fact, these methods produce results that are consistently worse than
  those of \(\mathrm{H}_\sigma\), and are not reproduced here.}

Even though we shall make some comparisons between the empirical
performance of the various methods on the real and the synthetic
datasets, these must of course be taken with caution, at least because
the latter dataset is much smaller than the former.

It turns out that also on synthetic data the \(\mathrm{BS_i}\) method
displays an outstanding performance, that is much better than what the
various other methods can achieve, consistently over all degrees of
Hermite polynomials considered and all quantiles of the error
distribution. The \(\mathrm{BS_i}\) method achieves better accuracy on
the synthetic dataset than on real data. This may be explained by the
fact that synthetic data are more ``regular'' than real data, in the
sense that the latter are more noisy, hence may have a more irregular
distribution. It is interesting also to observe that, for synthetic
data, accuracy within the hull is much better than the accuracy on the
whole dataset (i.e. including out-of-the-hull points). This points to
the plausibility of the previous argument, in the sense that
regularity of the distribution of synthetic data implies that
estimates in the hull are particularly precise.
On the other hand, the ``naive'' Black-Scholes estimator BS performs
considerably worse on synthetic data than on real data. A possible
explanation for this is that Hermite processes of order three, as the
one used to generate the data, are strongly non-Gaussian. In a
somewhat loose way, one may argue that the non-Gaussianity of the
Hermite process used here is stronger than the non-Gaussianity of
returns in real data.

Method \(\mathrm{H}_\sigma\) produces estimates that are considerably
poorer than those produced by \(\mathrm{BS_i}\), in analogy with the
corresponding results for the real dataset. On the other hand, the
accuracy improves considerably with respect to the BS method, showing
that the Hermite approximation method captures deviations from
Gaussianity to a certain extent. Note also that there is essentially
no improvement passing from \(N=4\) to \(N=5\). It should also be
mentioned that the method performs worse on the synthetic data than on
the real data with \(N \leq 3\), while with \(N=4,5\) the performance
is very similar in the hull, but still worse (for the synthetic data)
out of the hull. This is probably still due to a stronger deviation
from Gaussianity in the synthetic data that cannot be captured
sufficiently well by Hermite approximations of the density of order up
to five.

Method \(\mathrm{H}^{1,0}_\sigma\) performs significantly worse than
the much quicker method \(H_\sigma\). As already remarked, the
optimization algorithm suffers from the existence of many local
minima, and the local minimum closest to the BS parameters, to which
it converges, may arguably be quite far from the global minimum. This
phenomenon was already observed in the case of real data, and it is
even more pronounced for synthetic data. It is perhaps worth noting
that the method has a median error that decreases as the order \(N\)
increases, but produces large errors that strongly influence the error
distributions at higher quantiles.

In contrast to \(\mathrm{H}^{1,0}_\sigma\), method
\(\mathrm{H}^{1,2}_\sigma\) searches for a local minimum, in the
\(\ell^1\) sense, starting from the parameters of
\(\mathrm{H}_\sigma\). This method, that could be seen as a refinement
of method \(\mathrm{H}_\sigma\), has an entirely similar accuracy to
that of the latter across all values of \(N\). Strictly speaking, this
may just be explained by the existence of a local minimum quite close
to the initial datum for the search algorithm. In practice, however,
in analogy to the case of real data, this shows that the much quicker
method \(\mathrm{H}_\sigma\), although theoretically not fully
consistent, produces estimates that can hardly be improved by standard
(non-global) optimization algorithms.
It appears interesting to observe that the median error is worse in
the synthetic data than in the real data, but that the frequency of
large errors, at least for sufficiently high order \(N\), is lower for
synthetic data than for real data. This might be consistent with real
data having a higher Gaussianity than synthetic data, but more
extreme outliers.

Finally, the extra parameter \(m\) allows method
\(\mathrm{H}_{m,\sigma}\) to achieve a higher accuracy than the
simpler method \(\mathrm{H}_\sigma\). This is of course not surprising
from the mere statistical viewpoint, but it may be somewhat
interesting nonetheless, considering how the synthetic data are
generated (i.e., roughly speaking, choosing \(m\) as in
\(\mathrm{H}_\sigma\)). This observation can be interpreted as further
evidence for a deviation from Gaussianity of Hermite processes that is
hard to capture with Hermite approximations (of order up to five, at
least).


\begin{table}[tbp]
\caption{Synthetic data: pricing errors Black and Scholes}
\begin{center}
  \vspace{2mm} \parbox{\textwidth}{\footnotesize The table reports
    selected quantiles of the distribution of empirical pricing errors
    (in percentage points) of the Black and Scholes estimators.
    Results are calculated on a set of synthetic data, generated using
    Hermite processes.  Each column matches the corresponding one in
    the tables relative to Hermite pricing. Figures in parenthesis
    refer to pricing errors obtained excluding strikes outside the
    interval of observed ones. \vspace{4mm}}
  \footnotesize{{\begin{center} 
 \begin{tabular*}{\textwidth}{@{\extracolsep{\fill}}lccccccc} 
 \toprule 
 \multicolumn{6}{c}{}\\ 
 \multicolumn{6}{c}{\textsc{Empirical distribution of Pricing errors}}\\ 
 \multicolumn{6}{c}{}\\ 
 \multicolumn{6}{c}{Black \& Scholes}\\ 
 \midrule 
Quantiles & $N=1$ &$N=2$ & $N=3$ & $N=4$ & $N=5$ \\  
 \cmidrule{2-6} 
$10\%$     & 3.4 (3.5)& 0.5 (0.6)& 0.4 (0.4)& 0.1 (0.1)& 0.1 (0.1)\\ 
$25\%$     & 9.5 (9.3)& 2.0 (1.9)& 1.6 (1.5)& 0.7 (0.6)& 0.4 (0.4)\\ 
$50\%$     & 21.1 (19.6)& 5.4 (5.0)& 5.0 (4.7)& 2.3 (2.1)& 1.8 (1.7)\\ 
$75\%$     & 47.7 (40.9)& 12.3 (10.7)& 11.5 (10.2)& 5.6 (4.7)& 5.0 (4.3)\\ 
$90\%$     & 78.2 (72.0)& 27.6 (19.8)& 24.6 (18.2)& 13.0 (8.8)& 11.4 (8.2)\\ 
$95\%$     & 88.1 (83.4)& 53.6 (30.7)& 46.9 (27.3)& 36.5 (13.9)& 27.6 (12.6)\\ 
 \multicolumn{6}{c}{}\\ 
 \multicolumn{6}{c}{Black \& Scholes with linearly interpolated $\sigma$}\\ 
 \midrule 
Quantiles & $N=1$ &$N=2$ & $N=3$ & $N=4$ & $N=5$ \\  
 \cmidrule{2-6} 
$10\%$     & 0.0 (0.0)& 0.0 (0.0)& 0.0 (0.0)& 0.0 (0.0)& 0.0 (0.0)\\ 
$25\%$     & 0.1 (0.1)& 0.1 (0.1)& 0.1 (0.1)& 0.1 (0.1)& 0.1 (0.1)\\ 
$50\%$     & 0.4 (0.3)& 0.4 (0.3)& 0.4 (0.3)& 0.3 (0.3)& 0.3 (0.3)\\ 
$75\%$     & 1.3 (1.1)& 1.3 (1.1)& 1.2 (1.0)& 1.2 (1.0)& 1.1 (1.0)\\ 
$90\%$     & 5.1 (3.5)& 4.8 (3.5)& 4.3 (3.2)& 3.9 (3.1)& 3.5 (3.0)\\ 
$95\%$     & 13.4 (7.7)& 12.5 (7.6)& 10.5 (7.1)& 9.7 (6.5)& 8.1 (5.7)\\ 
 \midrule 
Test points& 4561 (3994)& 4501 (3974)& 4401 (3924)& 4301 (3864)& 4187 (3788)\\ 
\bottomrule 
\end{tabular*} 
\end{center}}}
\end{center}
\label{tab:bsps}
\end{table}

\begin{table}[tbp]
  \caption{Synthetic data: errors for Hermite models
    $\mathrm{H}_\sigma$ and $\mathrm{H}_{m,\sigma}$}
\begin{center}
  \vspace{2mm} \parbox{\textwidth}{\footnotesize The table reports
    selected quantiles of the distribution of empirical estimation
    errors (in percentage points) of the Hermite estimators
    $\mathrm{H}_\sigma$ and $\mathrm{H}_{m,\sigma}$.  Results are
    calculated on a set of synthetic data, generated using Hermite
    processes.  Figures in parenthesis refer to estimation errors
    obtained excluding values of $k$ outside the interval of observed
    ones. \vspace{4mm}} \footnotesize{{\begin{center} 
 \begin{tabular*}{\textwidth}{@{\extracolsep{\fill}}lccccc} 
 \toprule 
 \multicolumn{6}{c}{}\\ 
 \multicolumn{6}{c}{\textsc{Empirical distribution of Pricing errors}}\\ 
 \multicolumn{6}{c}{}\\ 
 \multicolumn{6}{c}{$H_{\sigma}$}\\ 
 \midrule 
Quantiles & $N=1$ &$N=2$ & $N=3$ & $N=4$ & $N=5$ \\  
 \cmidrule{2-6} 
$10\%$     & 5.0 (5.1)& 0.5 (0.5)& 0.8 (0.8)& 0.4 (0.4)& 0.1 (0.1)\\ 
$25\%$     & 13.6 (13.0)& 1.8 (1.7)& 2.9 (2.9)& 1.3 (1.2)& 0.5 (0.5)\\ 
$50\%$     & 27.4 (25.2)& 7.2 (6.6)& 9.3 (8.8)& 3.2 (3.0)& 2.7 (2.5)\\ 
$75\%$     & 58.7 (51.7)& 19.4 (17.4)& 20.2 (18.6)& 7.5 (6.5)& 7.3 (6.5)\\ 
$90\%$     & 86.0 (80.8)& 47.9 (34.5)& 42.5 (32.8)& 15.0 (11.7)& 15.1 (11.8)\\ 
$95\%$     & 94.1 (89.7)& 73.2 (57.7)& 65.6 (50.9)& 33.5 (17.3)& 33.4 (17.8)\\ 
 \multicolumn{6}{c}{}\\ 
 \multicolumn{6}{c}{$H_{m,\sigma}$}\\ 
 \midrule 
Quantiles & $N=1$ &$N=2$ & $N=3$ & $N=4$ & $N=5$ \\  
 \cmidrule{2-6} 
$10\%$     & 3.4 (3.5)& 0.5 (0.6)& 0.4 (0.4)& 0.1 (0.1)& 0.1 (0.1)\\ 
$25\%$     & 9.5 (9.3)& 2.0 (1.9)& 1.6 (1.5)& 0.7 (0.6)& 0.4 (0.4)\\ 
$50\%$     & 21.1 (19.6)& 5.4 (5.0)& 5.0 (4.7)& 2.3 (2.1)& 1.8 (1.7)\\ 
$75\%$     & 47.7 (40.9)& 12.3 (10.7)& 11.5 (10.2)& 5.6 (4.7)& 5.0 (4.3)\\ 
$90\%$     & 78.2 (72.0)& 27.6 (19.8)& 24.6 (18.2)& 13.0 (8.8)& 11.4 (8.2)\\ 
$95\%$     & 88.1 (83.4)& 53.6 (30.7)& 46.9 (27.3)& 36.5 (13.9)& 27.6 (12.6)\\ 
\midrule 
Test points& 4561 (3994)& 4501 (3974)& 4401 (3924)& 4301 (3864)& 4187 (3788)\\ 
\bottomrule 
\end{tabular*} 
\end{center}
}}
\end{center}
\label{tab:novsim}
\end{table}

\begin{table}[tbp]
  \caption{Synthetic data: errors of Hermite models
    $\mathrm{H}_{\sigma}^{1,0}$ and $H_{\sigma}^{1,2}$}.
\begin{center}
  \vspace{2mm} \parbox{\textwidth}{\footnotesize The table reports
    selected quantiles of the distribution of empirical estimation
    errors (in percentage points) of the Hermite estimators
    $\mathrm{H}_{\sigma}^{1,0}$ and $H_{\sigma}^{1,2}$. Results are
    calculated on a set of synthetic data, generated using Hermite
    processes.  Figures in parenthesis refer to estimation errors
    obtained excluding values of $k$ outside the interval of observed
    ones.  \vspace{4mm}} \footnotesize{{\begin{center} 
 \begin{tabular*}{\textwidth}{@{\extracolsep{\fill}}lccccccc} 
 \toprule 
 \multicolumn{6}{c}{}\\ 
 \multicolumn{6}{c}{\textsc{Empirical distribution of Pricing errors}}\\ 
 \multicolumn{6}{c}{}\\ 
 \multicolumn{6}{c}{$H_{\sigma}^{1,0}$}\\ 
 \midrule 
Quantiles & $N=1$ &$N=2$ & $N=3$ & $N=4$ & $N=5$ \\  
 \cmidrule{2-6} 
$10\%$     & 2.5 (2.7)& 1.2 (1.2)& 0.7 (0.7)& 0.7 (0.7)& 0.9 (0.9)\\ 
$25\%$     & 7.3 (7.2)& 4.0 (3.9)& 2.5 (2.4)& 2.5 (2.4)& 2.5 (2.3)\\ 
$50\%$     & 20.7 (18.9)& 11.4 (10.4)& 8.8 (8.1)& 8.2 (7.6)& 7.2 (6.7)\\ 
$75\%$     & 74.4 (69.3)& 26.1 (22.4)& 24.6 (21.5)& 27.5 (23.1)& 21.3 (18.3)\\ 
$90\%$     & 94.8 (94.0)& 59.4 (46.9)& 69.6 (55.3)& 78.7 (71.2)& 72.8 (61.8)\\ 
$95\%$     & 97.4 (97.1)& 83.6 (73.6)& 88.0 (83.3)& 91.1 (89.0)& 89.8 (87.5)\\ 
 \multicolumn{6}{c}{}\\ 
 \multicolumn{6}{c}{$H_{\sigma}^{1,2}$}\\ 
 \midrule 
Quantiles & $N=1$ &$N=2$ & $N=3$ & $N=4$ & $N=5$ \\  
 \cmidrule{2-6} 
$10\%$     & 3.4 (3.5)& 0.6 (0.6)& 0.8 (0.8)& 0.4 (0.4)& 0.1 (0.1)\\ 
$25\%$     & 8.9 (8.5)& 1.9 (1.8)& 3.0 (3.0)& 1.3 (1.3)& 0.6 (0.5)\\ 
$50\%$     & 24.2 (21.1)& 6.7 (6.3)& 9.2 (8.8)& 3.3 (3.0)& 2.8 (2.6)\\ 
$75\%$     & 63.7 (56.4)& 19.5 (17.1)& 20.4 (18.6)& 7.6 (6.7)& 7.5 (6.6)\\ 
$90\%$     & 88.6 (84.1)& 50.0 (37.3)& 44.6 (34.2)& 15.3 (11.9)& 15.6 (12.2)\\ 
$95\%$     & 95.2 (92.4)& 74.9 (59.5)& 67.7 (54.6)& 33.7 (17.2)& 33.4 (18.6)\\ 
\midrule 
Test points& 4561 (3994)& 4501 (3974)& 4401 (3924)& 4301 (3864)& 4187 (3788)\\ 
\bottomrule 
\end{tabular*} 
\end{center}}}
\end{center}
\label{tab:agnsim}
\end{table}


\section{Concluding remarks}
We have analyzed the empirical performance of a class of
nonparametric models to price European options with fixed time to
maturity, based on approximating the density of logarithmic returns by
truncated series of weighted and scaled Hermite polynomials. As a term
of comparison we considered a simple Black-Scholes model coupled with
linear interpolation on the implied volatility curve.
The empirical performance of the methods, measured by off-sample
relative pricing error, is studied on one year of daily European put
options on the S\&P500 index. The results suggest that Hermite models
performs reasonably well for options with strike price not too far
away from the strike prices of observed option prices. This appears to
be the case across all calibration methods used. For options with
strike price far apart from the set of strike prices of observed
options, estimates obtained by the Hermite methods are less reliable
than those obtained by the simple nonparametric Black-Scholes method
mentioned above, and are in general not better otherwise. Therefore it
seems fair to say that Hermite methods can be useful in particularly
well-behaved situations, but cannot be considered reliable stand-alone
nonparametric pricing tools.
These qualitative observations are confirmed by a statistical exercise
conducted on synthetic data generated in terms of a class of
non-Gaussian stochastic processes (the Hermite processes).

\appendix

\section{Pricing formulas under extra integrability conditions}
\label{ssec:Hp2}
Assume that there exists $\delta>0$ such that
\[
  \widetilde{f} \colon x \mapsto e^{\sigma(1+\delta)\abs{x}}f(x) \in L^2.
\]
Let $(\widetilde{f}_N)$ be a sequence of function converging to
$\widetilde{f}$ in $L^2$ and define $(f_N)$ by
\[
  e^{\sigma(1+\delta)\abs{x}} f_N(x) = \widetilde{f}_N
  \qquad \forall N \geq 0.
\]
Lemma~\ref{lm:sigma} implies that there exists a
sequence $(\alpha_n) \in \ell^2$ such that
\[
  e^{\sigma(1+\delta)\abs{x}} f(x) = \sum_{n=0}^\infty
  \alpha_n h_n(\sqrt{2}x) e^{-x^2/2},
\]
therefore, setting
\[
  f_N(x) = \sum_{n=0}^N \alpha_n h_n(\sqrt{2}x)
  e^{-x^2/2 - \sigma(1+\delta)\abs{x}}
\]
for every $N \geq 0$, the sequence of functions defined by
$x \mapsto e^{\sigma(1+\delta)\abs{x}} (f_N(x)-f(x))$ converges to
zero in $L^2$ as $N \to \infty$.
Setting
\[
\zeta_+ := \frac{1}{\sigma} \Bigl(
\log \frac{k}{S_0} - m + \overline{q}t \Bigr),
\]
one has
\[
\E{(k-S_t)}^+ = \int_{-\infty}^{\zeta_+} \bigl( k - e^{\sigma x +
m - \overline{q}t} \bigr) f(x)\,dx,
\]
hence, approximating $f$ by $f_N$, we define
\begin{align*}
  \pi^N &:= \int_{-\infty}^{\zeta_+} \bigl( k - S_0 e^{\sigma x +
    m - \overline{q}t} \bigr) f_N(x)\,dx\\
  &= k \int_{-\infty}^{\zeta_+} f_N(x)\,dx - e^{-\overline{q}t}
  S_0\int_{-\infty}^{\zeta_+} f_N(x) e^{\sigma x + m} \,dx.
\end{align*}
We have
\[
  k \int_{-\infty}^{\zeta_+} f_N(x)\,dx = \sum_{n=0}^N \alpha_n k
  \int_{-\infty}^{\zeta_+} h_n(\sqrt{2}x) e^{-x^2/2 -
    \sigma(1+\delta)\abs{x}}\,dx,
\]
where, if $\zeta_+ \leq 0$, setting $\sigma_\delta:=\sigma(1+\delta)$
for notational compactness, \eqref{eq:inta} implies
\begin{align*}
\int_{-\infty}^{\zeta_+} h_n(\sqrt{2}x) e^{-x^2/2 - \sigma_\delta\abs{x}}\,dx
&= \int_{-\infty}^{\zeta_+} h_n(\sqrt{2}x) e^{-x^2/2 + \sigma_\delta x}\,dx\\
&= e^{\sigma_\delta^2/2} \int_{-\infty}^{\zeta_+-\sigma_\delta} 
     h_n(\sqrt{2}(x+\sigma_\delta)) e^{-x^2/2}\,dx.
\end{align*} 
Similarly, if $\zeta_+ \geq 0$, analogous computations yield
\begin{align*}
&\int_{-\infty}^{\zeta_+} h_n(\sqrt{2}x) e^{-x^2/2 - \sigma_\delta\abs{x}}\,dx\\
&\hspace{3em} = \int_{-\infty}^0 h_n(\sqrt{2}x) e^{-x^2/2 + \sigma_\delta x}\,dx
+ \int_0^{\zeta_+} h_n(\sqrt{2}x) e^{-x^2/2 - \sigma_\delta x}\,dx\\
&\hspace{3em} = e^{\sigma_\delta^2/2} \int_{-\infty}^{-\sigma_\delta}
h_n(\sqrt{2}(x+\sigma_\delta)) e^{-x^2/2}\,dx\\
&\hspace{4em} + e^{\sigma_\delta^2/2}
\int_{\sigma_\delta}^{\zeta_++\sigma_\delta} h_n(\sqrt{2}(x-\sigma_\delta))
e^{-x^2/2}\,dx.
\end{align*}
Moreover,
\[
e^{m-\overline{q}t} S_0 \int_{-\infty}^{\zeta_+} f_N(x) e^{\sigma x} \,dx
= \sum_{n=0}^N \alpha_n e^{m-\overline{q}t} S_0
\int_{-\infty}^{\zeta_+} h_n(\sqrt{2}x) e^{-x^2/2 - \sigma_\delta\abs{x} + \sigma x}\,dx,
\]
where, if $\zeta_+ \leq 0$, noting that $\sigma_\delta+\sigma =
2\sigma_{\delta/2}$,
\begin{align*}
  \int_{-\infty}^{\zeta_+} h_n(\sqrt{2}x) e^{-x^2/2 - \sigma_\delta\abs{x} + \sigma x}\,dx
  &= \int_{-\infty}^{\zeta_+} h_n(\sqrt{2}x) e^{-x^2/2 + 2\sigma_{\delta/2}x}\,dx\\
  &= e^{2\sigma_{\delta/2}^2} \int_{-\infty}^{\zeta_+-2\sigma_{\delta/2}} 
  h_n(\sqrt{2}(x+2\sigma_{\delta/2})) e^{x^2/2} \,dx,
\end{align*}
and, if $\zeta_+ \geq 0$,
\begin{align*}
&\int_{-\infty}^{\zeta_+} h_n(\sqrt{2}x) e^{-x^2/2 - \sigma_\delta\abs{x} + \sigma x}\,dx\\
&\hspace{3em}= \int_{-\infty}^0 h_n(\sqrt{2}x) e^{-x^2/2 + 2\sigma_{\delta/2}x}\,dx
+ \int_0^{\zeta_+} h_n(\sqrt{2}x) e^{-x^2/2 - \sigma\delta x}\,dx\\
&\hspace{3em}= e^{2\sigma_{\delta/2}^2} \int_{-\infty}^{-2\sigma_{\delta/2}}
h_n(\sqrt{2}(x+2\sigma_{\delta/2})) e^{x^2/2} \,dx
+ e^{\sigma^2\delta^2/2} \int_{\sigma\delta}^{\zeta_+ + \sigma\delta}
h(\sqrt{2}(x-\sigma\delta)) e^{-x^2/2}\,dx.
\end{align*}

\section{Hermite processes: weak convergence and simulation}
\label{app:HermProc}
We collect some facts about Hermite processes, using as main source
\cite{Taqqu:Herm} (see also \cite{DoMa}).

Let us first define (the class of) Hermite processes. To this purpose,
we need to fix some notation. Throughout this section,
\(k \in \enne\), \(k \geq 1\) and \(H \in \mathopen]1/2,1\mathclose[\)
are constants, and \(L\) is a slowly varying function at infinity that
is bounded on bounded intervals of
\(\mathopen]0,+\infty\mathclose[\). Furthermore, let
\[
  H_0 := 1 - \frac{1-H}{k},
\]
or, equivalently, \(H = k(H_0-1)+1\) (cf.~\cite[(1.7)]{Taqqu:Herm}),
and define the constant \(c = c(k,H_0)\) by
\begin{align*}
  c
  &= \biggl( \frac{k!(k(H_0-1)+1)(2k(H_0-1)+1)}%
    {\bigl(\int_0^\infty(u+u^2)^{H_0-3/2}du\bigr)^k} \biggr)^{1/2} \\
  &= \bigl( k!(k(H_0-1)+1)(2k(H_0-1)+1) \bigr)^{1/2}
    \biggl( \frac{\Gamma(3/2-H_0)}{\Gamma(H_0-1/2) \Gamma(2-2H_0)}
    \biggr)^{k/2}\\
  &= \bigl( k! H (2H-1) \bigr)^{1/2}
  \left( \frac{\Gamma(1/2 + (1-H)/k)}%
    {\Gamma(1/2 - (1-H)/k) \, \Gamma(2(1-H)/k)} \right)^{k/2}
\end{align*}
(cf.~\cite[(1.6)]{Taqqu:Herm}).

The Hermite process with parameters \(k\) and \(H\) is defined by
\[
  Z^k_H (t)= c \int_{\erre^k} \int_0^t
  \prod_{j=1}^{k} (s-y_i)_+^{-\left( \frac{1}{2} + \frac{1-H}{k}
      \right)} \,ds\,dW(y_{1})\, \cdots \,dW(y_{k}),
\]
where \(x_{+}\) denotes the positive part of $x$ and the integral is a
multiple Wiener-It\^o stochastic integral with respect to a Wiener
process \(W\) with parameter space $\erre$.  Then \(Z^k_H\) is a
mean-zero square-integrable process with stationary increments,
\(Z^k_H(0)=0\), and \(\E(Z^k_H(1))^2 = 1\). Moreover, \(Z^k_H\) is
self-similar with parameter \(H\), i.e. \(Z^k_H(t) = t^H Z^k_H(1)\) in
distribution for every \(t \in \erre_+\). Finally, for $k\geq 2$ the
process is not Gaussian.

\medskip

Let \(g \in L^2(\gamma)\) be a function such that
\(\gamma(g)=0\). Then \(g\) can be written as
\[
  g(x) = \sum _{j \geq 1} \alpha_j H_j(x), \qquad
  \alpha_j = \frac{1}{j!} \ip[\big]{g}{H_j}_{L^2(\gamma)}.
\]
%
The function \(g\) is said to have Hermite rank $k$ if
\(c_0=\cdots=c_{k-1}=0,\, c_k \neq 0\).  Since \(\gamma(g) = 0\), one
has \(k \geq 1\).

\medskip

Taqqu has proved in \cite[Theorem~5.5]{Taqqu:Herm} a convergence
theorems for integral functionals of a class of Gaussian processes
\(X\) that admits a representation of the type
\begin{equation}
  \label{eq:X}
  X_t = \frac{1}{\sigma}\int_\erre e(t-s)\,dW_s,
\end{equation}
where \(e\colon \erre \to \erre\) is a function satisfying a set of
conditions spelled out in \cite[{\S}2]{Taqqu:Herm}, that include
\[
\lim_{s \to \infty} \frac{e(s)}{s^{H_0-3/2}L(s)} = 1
\]
and \(\sigma := {\norm{e}}_{L^2(\erre)}\). Then one has, for any
\(g \in L^2(\gamma)\) of Hermite rank \(k\),
\[
  \lim_{x \to \infty} \frac{1}{A(x)} \int_0^{xt} g(X_s)\,ds
  = \alpha_k Z^k_H(t)
\]
in the sense of weak convergence of measures in \(C[0,1]\). Here
\(A(x)\) is a normalizing constant defined by
\[
  A(x) := \frac{k!}{\sigma^k c(k,H_0)} x^H L^k(x).
\]

\medskip

The process \(B^H:=Z^1_H\) is the usual fractional Brownian
motion. The centered stationary Gaussian process \(X\) defined by
\(X_t := B^H_t - B^H_{t-1}\) is called fractional Gaussian noise.  The
process \(X\) admits a representation of the type \eqref{eq:X}, with
\[
  e(t) = t^{H-3/2} L(t), \qquad
  L(t) =
  \begin{cases}
    0, & t \in \erre_-,\\
    t, & t \in [0,1],\\
    t^{3/2-H} \bigl( t^{H-1/2} - (t-1)^{H-1/2} \bigr),
      &t \in [1,\infty\mathclose[,
  \end{cases}
\]
for which then
\[
  \sigma = \frac{H-1/2}{c(1,H)},
\]
(see \cite[{\S}3]{Taqqu:Herm}).
As a consequence,
\begin{equation}
  \label{eq:serie}
  \lim_{N \to \infty} \frac{1}{A(N)} \sum_{i=1}^{[Nt]} g(X_i)
  = \alpha_k Z^k_H(t)
\end{equation}
in the sense of weak convergence of measures in the Skorokhod space
\(D[0,1]\) (cf. \cite[Theorem~5.6]{Taqqu:Herm}).

\smallskip

The convergence result \eqref{eq:serie} is amenable to practical
implementation, as the (well-known) covariance function of \(B^H\) is
\[
  (t_1,t_2) \mapsto \frac12 \Bigl( t_1^{2H} + t_2^{2H} - \abs{t_1-t_2}^{2H}
  \Bigr),
\]
hence, by an elementary computation, the (stationary) correlation
function of the corresponding fractional Gaussian noise \(X\) is given
by
\begin{equation*}
  \label{eq:rho}
  \rho(m) := \E X_tX_{t+m} = \frac12 \Bigl( (\abs{m}+1)^{2H} +
  \abs[\big]{\abs{m}-1}^{2H} - 2\abs{m}^{2H} \Bigr), \qquad m \in \erre.
\end{equation*}


\section{On the density of a cubic function of a Gaussian}
\label{app:dens}
Let \(h(x) = x^3 - 3x\) be the Hermite function of order three.
Setting
\[
  h_1 = h\big\vert_{\mathopen]-\infty,-1\mathclose[},
  \qquad h_2 = h\big\vert_{\mathopen]-1,1\mathclose[},
  \qquad h_3 = h\big\vert_{\mathopen]1,\infty\mathclose[},
\]
it is immediate to see that
\begin{itemize}
\item[(i)] \(h_1\) is a strictly increasing \(C^\infty\) homeomorphism
  of \(\mathopen]-\infty,-1\mathclose[\) to
  \(\mathopen]-\infty,2\mathclose[\);
\item[(ii)] \(h_2\) is a strictly decreasing \(C^\infty\)
  homeomorphism of \(\mathopen]-1,1\mathclose[\) to
  \(\mathopen]-2,2\mathclose[\);
\item[(iii)] \(h_3\) is a strictly increasing \(C^\infty\)
  homeomorphism of \(\mathopen]1,\infty\mathclose[\) to
  \(\mathopen]-2,\infty\mathclose[\),
\end{itemize}
and that the function \(h\) has a local maximum at \(-1\) and a local
minimum at \(1\). The inverse of the function \(h_j\), \(j=1,2,3\),
will be denoted by \(h_j^\leftarrow\).

Let \(Z\) a standard Gaussian random variable. Then the distribution
function of the random variable \(h(Z)\) can be written as
\[
G(y) :=
\begin{cases}
  \P(Z \leq h_1^\leftarrow(y)), & y \leq -2,\\[3pt]
  \P(Z \leq h_1^\leftarrow(y)) + \P(Z \leq h_3^\leftarrow(y))
  - \P(Z \leq h_2^\leftarrow(y)), & \abs{y} < 2,\\[3pt]
  \P(Z \leq h_3^\leftarrow(y)), & y \geq 2,
\end{cases}
\]
that is, denoting the distribution function of \(Z\) by \(\Phi\),
\[
G(y) :=
\begin{cases}
  \Phi \circ h_1^\leftarrow(y), & y \leq -2,\\[3pt]
  \Phi \circ h_1^\leftarrow(y) + \Phi \circ  h_3^\leftarrow(y)
  - \Phi \circ h_2^\leftarrow(y), & \abs{y} < 2,\\[3pt]
  \Phi \circ h_3^\leftarrow(y), & y \geq 2.
\end{cases}
\]
The inverse function theorem readily shows that the limits of
\(\bigl( h_1^\leftarrow \bigr)'\) and
\(\bigl( h_2^\leftarrow \bigr)'\) at \(2\), and the limits of
\(\bigl( h_2^\leftarrow \bigr)'\) and
\(\bigl( h_3^\leftarrow \bigr)'\) at \(-2\), are infinite, hence it is
not clear whether the square of the derivative of \(G\) is integrable
on neighborhoods of \(2\) and \(-2\).
To answer this question, note that the cubic equation
\(x^3 - 3x = y\), with \(y \in \mathopen]-2,2\mathclose[\), admits the
three real solutions
\[
  x_k = 2 \cos\Bigl( \frac13 \arccos y/2 - \frac{2\pi}{3} k \Bigr),
  \qquad k=0,1,2.
\]
In particular, there exists \(j \in \{1,2,3\}\) such that
\[
h_j^{\leftarrow}(y) = 2 \cos\Bigl( \frac13 \arccos y/2 \Bigr),
\]
for which
\[
  \bigl(h_j^{\leftarrow}\bigr)'(y) = \frac23 \sin\Bigl(
  \frac13 \arccos y/2 \Bigr) \frac{1}{\sqrt{1-y^2/4}}.
\]
This in turn implies
\[
  G'(y) = \frac{1}{\sqrt{2\pi}} \exp \bigl( - (h_j^\leftarrow(y))^2/2 \bigr)
          \, \bigl(h_j^{\leftarrow}\bigr)'(y),
\]
where \(\lim_{y \to 2} h_j^\leftarrow(y)=1\) and
\[
\bigl(h_j^{\leftarrow}\bigr)'(y) = \frac{2}{(2-y)^{1/2}(2+y)^{1/2}},
\]
hence \(G'(y)^2\) tends to infinity as \(y \to 2-\) as \(1/(2-y)\),
which is not integrable. This implies that the (unbounded) density of
\(h(Z)\) does not belong to \(L^2(\erre)\).


\section{Numerical implementation}
\label{app:num}
All numerical computations are done with Octave 7.1.0 on Linux, using
the Octave Forge packages \texttt{statistics} and
\texttt{optim}. Minimization with respect to $\sigma$ in
$\mathrm{H}_\sigma$ is done through the function \texttt{fminbnd},
with lower and upper bounds $0.1$ and $1$, respectively. Minimizations
in $\mathrm{H}_\sigma^{1,0}$ and $\mathrm{H}_\sigma^{1,2}$ are done
through the function \texttt{fminsearch}.

Minimization in $\mathrm{H}_\sigma^1$ is done through the function
\texttt{glpk}, i.e. through the GNU Linear Programming Kit (GLPK). In
about 10\% of the computations it returns no solution. For these
points the result produced by $\mathrm{H}_\sigma^{1,2}$ is used
instead. Without any constraint on $\alpha$, the GLPK algorithm
sometimes breaks down or returns very high values of $\alpha$ that
translate into unusable pricing estimates. By trial-and-error, we
determined that a reasonable bound on the absolute value of $\alpha$
is $10$, and we implemented such a constraint.
The issues with minimization via GLPK are much more severe in the case
of $\mathrm{H}_{m,\sigma}^1$. Simple bounds on the absolute value of
$\alpha$ do not help in this case. The problem appears to be the
``size'' of the matrix $\Psi$, which depends on the parameters $m$ and
$\sigma$. Sporadic crashes of \texttt{glpk} (as opposed to very
frequent ones) are obtained by running it conditional on the absolute
value of the determinant of $\Psi^\top\Psi$ being bounded by
$10^6$. However, the proportion of pricing estimates obtained this way
is only around 5\% of the total.


\bibliographystyle{amsplain}
\bibliography{ref, finanza}

\end{document}